\documentclass[reqno,11pt]{amsart}
\usepackage[utf8]{inputenc}
\usepackage{graphicx}
\usepackage{slashed}
\usepackage{amscd}
\usepackage{amssymb}
\usepackage{esint}
\usepackage{enumitem}
\usepackage[mathscr]{eucal}
\numberwithin{equation}{section}
\allowdisplaybreaks[1]

\title[Existence of Minimizers in the Non-compact Setting]{Causal Variational Principles in the
$\sigma$-Locally Compact Setting: Existence of Minimizers}

\author[F.\ Finster]{Felix Finster}
\author[C.\ Langer]{Christoph Langer \\ \\ February / September 2020}
\address{Fakult\"at f\"ur Mathematik \\ Universit\"at Regensburg \\ D-93040 Regensburg \\ Germany}
\email{finster@ur.de, christoph.langer@ur.de}

\newtheorem{Def}{Definition}[section]
\newtheorem{Thm}[Def]{Theorem}
\newtheorem{Prp}[Def]{Proposition}
\newtheorem{Lemma}[Def]{Lemma}

\newtheorem{Corollary}[Def]{Corollary}

\newcommand{\Thanks}{\vspace*{.5em} \noindent \thanks}
\newcommand{\beq}{\begin{equation}}
\newcommand{\eeq}{\end{equation}}
\newcommand{\Proof}{\begin{proof}}
	\newcommand{\QED}{\end{proof} \noindent}

\newcommand{\C}{\mathbb{C}}
\newcommand{\R}{\mathbb{R}}

\newcommand{\N}{\mathbb{N}}

\renewcommand{\L}{{\mathcal{L}}}
\newcommand{\LL}{{\textup{L}}}
\newcommand{\Sact}{{\mathcal{S}}}

\newcommand{\M}{{\mathfrak{M}}}
\newcommand\B{{\mathscr{B}}}

\renewcommand{\H}{\mathscr{H}}

\newcommand{\F}{{\mathscr{F}}}

\DeclareMathOperator{\supp}{supp}

\newcommand{\s}{\mathfrak{s}}

\newcommand{\bitem}{\begin{itemize}[leftmargin=2.5em]}
\newcommand{\eitem}{\end{itemize}}

\DeclareFontFamily{OT1}{rsfso}{}
\DeclareFontShape{OT1}{rsfso}{m}{n}{ <-7> rsfso5 <7-10> rsfso7 <10-> rsfso10}{}
\DeclareMathAlphabet{\mycal}{OT1}{rsfso}{m}{n}

\begin{document}

\maketitle

%\vspace*{0.3cm}

\begin{abstract}
We prove the existence of minimizers of causal variational principles on second countable, locally compact Hausdorff spaces. Moreover, the corresponding Euler-Lagrange equations are derived.
The method is to first prove the existence of minimizers of the causal variational principle
restricted to compact subsets for a lower semi-continuous Lagrangian.
Exhausting the underlying topological space by compact subsets and rescaling the corresponding minimizers,
we obtain a sequence which converges vaguely to a regular Borel measure of possibly infinite total volume.
It is shown that, for continuous Lagrangians of compact range,
this measure solves the Euler-Lagrange equations. Furthermore, we prove that the constructed measure
is a minimizer under variations of compact support. Under additional assumptions, it is proven that
this measure is a minimizer under variations of finite volume. We finally extend our results to continuous Lagrangians decaying in entropy.
\end{abstract}

\vspace*{0.75cm}

\tableofcontents

\thispagestyle{empty} %UnterdrÃŒckt die Seitenzahl
%\newpage

\section{Introduction} \label{Subsection Overview}
In the physical theory of causal fermion systems, 
spacetime and the structures therein are described by a minimizer
of the so-called causal action principle
(for an introduction to the physical background and the mathematical context, we refer the interested reader to \S\ref{Subsection Motivation}, the textbook~\cite{cfs} and the survey articles~\cite{dice2014, %nrstg, 
review}). 
{\em{Causal variational principles}} evolved as a mathematical generalization
of the causal action principle~\cite{continuum, jet}.
The starting point in~\cite{jet} is a smooth manifold~$\F$
and a non-negative function~$\L : \F \times \F \rightarrow \R^+_0 := [0, \infty)$ (the {\em{Lagrangian}}) which is assumed
to be lower semi-continuous.
The causal variational principle is to minimize the {\em{action}}~$\Sact$ defined as
the double integral over the Lagrangian
\[ \Sact (\rho) = \int_\F d\rho(x) \int_\F d\rho(y)\: \L(x,y) \]
under variations of the measure~$\rho$ within the class of regular Borel measures,
keeping the total volume~$\rho(\F)$ fixed ({\em{volume constraint}}).
The aim of the present paper is to extend the existence theory for minimizers of such variational principles to the case that~$\F$ is non-compact
and the total volume is infinite. Furthermore, we drop the manifold structure of the underlying space~$\F$
and consider a $\sigma$-locally compact topological space instead.
We also work out the corresponding Euler-Lagrange (EL) equations.

In order to put the paper into the mathematical context, in~\cite{pfp} it was proposed to
formulate physics by minimizing a new type of variational principle in spacetime.
The suggestion in~\cite[Section~3.5]{pfp} led to the causal action principle
in discrete spacetime, which was first analyzed mathematically in~\cite{discrete}.
A more general and systematic inquiry of causal variational principles on measure spaces was carried out
in~\cite{continuum}. In this article, the existence of minimizers is proven in the case that the total volume is finite. 
In~\cite{jet}, the setting is generalized to non-compact manifolds of possibly infinite volume
and the corresponding EL equations are analyzed. However, the existence of minimizers is not proved.
Here we fill this gap and develop the existence theory in the non-compact setting.

The main difficulty in dealing with measures of infinite total volume is
to properly implement the volume constraint. 
Indeed, the naive prescription~$\rho(\F)=\infty$ leaves the freedom to
change the total volume by any finite amount, which is not sensible.
The way out is to only allow for variations which leave the measure unchanged
outside a set of finite volume (so-called {\em{variations of finite volume}}; see Definition~\ref{deffv}).
In order to prove existence of minimizers within this class, we exhaust~$\F$
by compact sets~$K_n$ and show that minimizers for the variational principle 
restricted to each~$K_n$ exist. Making essential use of the corresponding EL equations,
we rescale the minimizing measures in such a way that a subsequence converges
vaguely to a measure~$\rho$ on~$\F$. We proceed by proving that this measure
satisfies the EL equations globally. Finally, we prove that, under suitable assumptions,
this measure is even a minimizer under variations of finite volume.
This minimizing property is proved in two steps: We first assume that the Lagrangian
is of {\em{compact range}} (see Definition~\ref{defcompactrange}) and prove that~$\rho$
is a minimizer under {\em{variations of compact support}} (see Definition~\ref{Def global} and
Theorem~\ref{Lemma compact minimizer}).
In a second step we extend this result to variations of finite volume (see Definition~\ref{deffv} and Theorem~\ref{Theorem global})
under the assumption that property~(iv) in~\S\ref{seccvp} holds, i.e.\
\[ \sup_{x \in \F} \int_{\F}\L(x,y)\: d\rho(y) < \infty \:. \]
Sufficient conditions for this assumption to hold are worked out (see Lemma~\ref{Lemma conditions}).
Finally, we generalize our results to Lagrangians which do not have compact range,
but instead have suitable decay properties (see Definition~\ref{Definition vanishing}
and Theorem~\ref{Theorem minimizer vanishing}).

The paper is organized as follows.
In Section~\ref{secbackground} we give a short physical motivation (\S\ref{Subsection Motivation}) and
recall the main definitions and existence results as worked out in \cite{jet} (\S\ref{seccvp}).
In Section~\ref{seccvpsigma} causal variational principles in the $\sigma$-locally compact setting
are introduced (\S\ref{S Basic Definitions}), and the existence of minimizers is proved for the causal variational principle
restricted to compact subsets, making use of the Banach-Alaoglu theorem and the Riesz representation theorem (\S\ref{S Existence compact}).
In Section~\ref{secmincr} minimizers are constructed for continuous Lagrangians of compact range.
To this end, in \S\ref{secconstglobal} we exhaust the underlying topological space by
compact subsets and take a vague limit of suitably rescaled minimizers thereon to obtain a 
regular Borel measure on the whole topological space.
In~\S\ref{Section Derivation} it is shown that this measure satisfies the EL equations.
Furthermore, we prove in~\S\ref{subsection compact variation} that this measure is a minimizer
under variations of compact support (see Definition \ref{Def global}). 
Finally, in~\S\ref{secvaryfinite} it is shown that, under additional assumptions, this measure
is also a minimizer under variations of finite volume (see Definition~\ref{Definition minimizer}). 
In Section~\ref{Chapter 5} we conclude the paper by weakening
the assumption that the Lagrangian is of compact range to Lagrangians which {\em{decay in entropy}} (see Definition \ref{Definition vanishing}). Then the EL equations are again satisfied, and under similar additional assumptions as before we prove that the constructed Borel measure is a minimizer of the causal action principle as intended in~\cite{jet}. 

\section{Physical Background and Mathematical Preliminaries} \label{secbackground}
\subsection{Physical Context and Motivation} \label{Subsection Motivation}
The purpose of this subsection is to outline a few concepts of causal fermion systems and
to explain how the present paper fits into the general physical context and the ongoing research program.
The reader not interested in the physical background may skip this section.

The theory of causal fermion systems is a recent approach to fundamental physics.
The original motivation was to resolve shortcomings of relativistic quantum field theory.
Namely, due to ultraviolet divergences, perturbative quantum field theory is well-defined only after regularization, which is usually understood as a set of prescriptions for how to make divergent integrals finite (e.g.~by introducing a suitable ``cutoff'' in momentum space). The regularization is then removed using the renormalization procedure. However, this concept is not convincing from neither the physical nor the mathematical point of view. More precisely, in view of Heisenberg's uncertainty principle, physicists infer a correspondence between large momenta and small distances. Because of that, the regularization length is often associated to the Planck length~$\ell_P \approx 1.6 \cdot 10^{-35}\;\text{m}$. Accordingly, by introducing an ultraviolet cutoff in momentum space, one disregards distances which are smaller than the Planck length. As a consequence, the microscopic structure of spacetime is completely unknown. Unfortunately, at present
there is no consensus on what the correct mathematical model for ``Planck scale physics'' should be.

The simplest and maybe most natural approach is to assume that on the Planck scale, spacetime is no longer a continuum but becomes in some way ``discrete.'' This is the starting point in the monograph~\cite{pfp},
where the physical system is described by an ensemble of wave functions in a discrete spacetime.
Motivated by the Lagrangian formulation of classical field theory,
physical equations are formulated by a variational principle in discrete spacetime.
In the meantime, this setting was generalized and developed to the theory of causal fermion systems.
It is an essential feature of the approach that spacetime does not enter the variational principle a-priori,
but instead it emerges when minimizing the action. Thus causal fermion systems allow for the 
description of both discrete and continuous spacetime structures.

In order to get the connection to the present paper, let us briefly outline the main structures of causal fermion systems.
As initially introduced in~\cite{rrev}, a \emph{causal fermion system} consists of a triple~$(\H, \F, \rho)$ together with an integer $n \in \N$, where~$\H$ denotes a complex Hilbert space, $\F \subset \LL(\H)$ being the set of all self-adjoint operators on $\H$ of finite rank with at most $n$ positive and at most $n$ negative eigenvalues, and~$\rho$, referred to as \emph{universal measure}, being a measure on the Borel $\sigma$-algebra over~$\F$. Then for any~$x,y \in \F$, the product~$xy$ is an operator of rank at most~$2n$. Denoting its non-trivial eigenvalues (counting algebraic
multiplicities) by~$\lambda_1^{xy}, \ldots, \lambda_{2n}^{xy} \in \C$, and introducing the spectral weight~$|.|$ of an operator as the sum of the absolute values of its eigenvalues, the \emph{Lagrangian} can be introduced as a mapping
\begin{align*}
\L : \F \times \F \to \R_0^+ \:, \qquad \L(x,y) = \left|(xy)^2\right| - \frac{1}{2n} \left|xy\right|^2 \:.
\end{align*}
As being of relevance for this article, we point out that the Lagrangian is a continuous function which is symmetric in the sense that 
\[\L(x,y) = \L(y,x) \qquad \text{for all $x,y \in \F$} \:. \]
In analogy to classical field theory, one defines the \emph{causal action} by
\begin{align*}
\Sact(\rho) = \iint_{\F \times \F} \L(x,y) \: d\rho(x) \: d\rho(y) \:.
\end{align*}
Finally, the corresponding {\em{causal action principle}} is introduced by varying the measure~$\rho$ in 
the class of a suitable class of Borel measures under additional constraints
(which assert the existence of non-trivial minimizers). Given a minimizing measure~$\rho$,
{\em{spacetime}}~$M$ is defined as its support,
\[ M := \supp \rho \:. \]
As being outlined in detail in~\cite{cfs}, critical points of the causal action give rise to Euler-Lagrange (EL) equations, which describe the dynamics of the causal fermion system.
In a certain limiting case, the so-called {\em{continuum limit}},
one gets a connection to the conventional formulation of physics in a spacetime continuum.
In this limiting case, the EL equations give rise to classical field equations like the Maxwell and Einstein equations. 
Moreover, quantum mechanics is obtained in a limiting case, and close connections to 
relativistic quantum field theory have been established (see~\cite{qft} and~\cite{fockbosonic}).

In order for the causal action principle to be mathematically sensible,
the existence theory is of crucial importance.
In the case that the dimension of~$\H$ is finite, the existence of minimizers was proven in~\cite[Section~2]{continuum}
(based on the existence theory in discrete spacetime~\cite{discrete}),
giving rise to minimizing measures of finite total volume~$\rho(\F)<\infty$.
The remaining open problem
is to extend the existence theory to the case that~$\H$ is infinite-dimensional.
Then the total volume~$\rho(\F)$ is necessarily infinite (for a counter example see~\cite[Exercise~1.3]{cfs}).
Proving existence of minimizers in the resulting {\em{infinite-dimensional setting}} (i.e.~$\dim \H = \infty$ and~$\rho(\F)=\infty$)
is a difficult task. Therefore, our strategy is to approach the problem in two steps.
The first step is to deal with infinite total volume; it is precisely the objective of the present paper to address this problem in sufficient generality.
The second step, which involves the difficulty of dealing with non-locally compact spaces,
is currently under investigation.

\subsection{Causal Variational Principles in the Non-Compact Setting} \label{seccvp}
Before introducing the $\sigma$-locally compact setting in Section~\ref{seccvpsigma},
we now recall known results in the slightly less general situation of 
causal variational principles in the non-compact setting as
studied in~\cite[Section~2]{jet}. The starting point in~\cite{jet}
is a (possibly non-compact)
smooth manifold~$\F$ of dimension~$m \geq 1$. We let~$\rho$ be a (positive) measure on
the Borel algebra of~$\F$ (the {\em{universal measure}}).
Moreover, let~$\L : \F \times \F \rightarrow \R^+_0$ be a non-negative function (the {\em{Lagrangian}}) with the
following properties:
\begin{itemize}[leftmargin=2em]
\item[(i)] $\L$ is symmetric, i.e.\ $\L(x,y) = \L(y,x)$ for all~$x,y \in \F$. \label{Cond1}%
\item[(ii)] $\L$ is lower semi-continuous, i.e.\ for all sequences~$x_n \rightarrow x$ and~$y_{n'} \rightarrow y$,
\label{Cond2}%
\[ \L(x,y) \leq \liminf_{n,n' \rightarrow \infty} \L(x_n, y_{n'})\:. \]
\end{itemize}
The {\em{causal variational principle}} is to minimize the action
\beq \label{Sact} 
\Sact (\rho) = \int_\F d\rho(x) \int_\F d\rho(y)\: \L(x,y) 
\eeq
under variations of the measure~$\rho$, keeping the total volume~$\rho(\F)$ fixed
({\em{volume constraint}}). Here we are interested in the case that the total volume is
infinite. In order to implement the volume constraint, we make the following additional assumptions:
\begin{itemize}[leftmargin=2em]
\item[(iii)] The measure~$\rho$ is {\em{locally finite}}
(meaning that any~$x \in \F$ has an open neighborhood~$U \subset \F$ with~$\rho(U)< \infty$). \label{Cond3}%
%and {\em{regular}} . 
\item[(iv)] The function~$\L(x,.)$ is $\rho$-integrable for all~$x \in \F$ and \label{Cond4}%
\beq \label{Lint}
\sup_{x \in \F} \int_{\F}\L(x,y)\: d\rho(y) < \infty \:.
\eeq
\end{itemize}
By Fatou's lemma, the integral in~\eqref{Lint} is lower semi-continuous in the variable~$x$.
A measure on the Borel algebra which satisfies~(iii) will be referred to as a {\em{Borel measure}}.
A Borel measure is said to be {\em{regular}} if it is inner and outer regular.

In order to give the causal variational principle~\eqref{Sact} a mathematical meaning, we 
first note that the difference of two measures $\rho, \tilde{\rho} : \B(\F) \to [0, + \infty]$ with~$\rho(\F), \tilde{\rho}(\F) = \infty$ is not a signed measure (due to expressions of the kind~``$\infty - \infty$''); as a consequence, the total variation for signed measures does not apply (see e.g.~\cite[\S6.1]{rudin}). For this reason, we introduce the total variation of the difference of two such measures by saying that~$|\rho - \tilde{\rho}|(\F) < \infty$ if and only if there exists
a Borel set~$B \subset \F$ with~$\rho(B), \tilde{\rho}(B) < \infty$ and~$\rho|_{\F \setminus B} = \tilde{\rho}|_{\F \setminus B}$. In this case, the signed
measure~$\rho-\tilde{\rho}$ for any Borel set~$\Omega \subset \F$ is defined by
\[(\rho-\tilde{\rho})(\Omega) := \rho(\Omega \cap B) - \tilde{\rho}(\Omega \cap B) \:. \]
We then vary in the following class of measures:
\begin{Def} \label{deffv}
Given a regular Borel measure~$\rho$ on~$\F$, a regular Borel measure~$\tilde{\rho}$ on~$\F$
is said to be a {\bf{variation of finite volume}} if 
\beq \label{totvol}
\big| \tilde{\rho} - \rho \big|(\F) < \infty \qquad \text{and} \qquad
\big( \tilde{\rho} - \rho \big) (\F) = 0 \:.
\eeq
\end{Def} \noindent

Assuming that (i), (ii) and (iv) hold 
and that~$\tilde{\rho}$ is a variation of finite volume, the difference of the actions as given by
\beq \label{integrals}
\begin{split}
\big( &\Sact(\tilde{\rho}) - \Sact(\rho) \big) = \int_\F d(\tilde{\rho} - \rho)(x) \int_\F d\rho(y)\: \L(x,y) \\
&\quad + \int_\F d\rho(x) \int_\F d(\tilde{\rho} - \rho)(y)\: \L(x,y) 
+ \int_\F d(\tilde{\rho} - \rho)(x) \int_\F d(\tilde{\rho} - \rho)(y)\: \L(x,y)
\end{split}
\eeq
is well-defined in view of~\cite[Lemma~2.1]{jet}.

\begin{Def} \label{Definition minimizer'} A regular Borel measure~$\rho$ for which the above conditions~{\rm{(i)--(iv)}} hold is said to be a {\bf{minimizer}} of the causal action
if the difference~\eqref{integrals}
is non-negative for all regular Borel measures~$\tilde{\rho}$ satisfying~\eqref{totvol}, i.e.\
\[ \big( \Sact(\tilde{\rho}) - \Sact(\rho) \big) \geq 0 \:. \]
\end{Def}

We denote the support of the measure~$\rho$ by~$M$,
\beq \label{suppdef}
M := \supp \rho = \F \setminus \bigcup \big\{ \text{$\Omega \subset \F$ \,\big|\,
$\Omega$ is open and $\rho(\Omega)=0$} \big\}
\eeq
(thus the support is the set of all points for which every open neighborhood
has a strictly positive measure; for details and generalizations
see~\cite[Subsection~2.2.1]{federer}).

It is shown in~\cite[Lemma~2.3]{jet} (based on a similar result
in the compact setting in~\cite[Lemma~3.4]{support}) that a minimizer
satisfies the following {\em{Euler-Lagrange (EL) equations}},
which state that for a suitable value of the parameter~$\s>0$,
the lower semi-continuous function~$\ell : \F \rightarrow \R_0^+$ defined by
\[ %\label{ldef}
\ell(x) := \int_\F \L(x,y)\: d\rho(y) - \s \]
is minimal and vanishes on the support of~$\rho$,
\beq \label{EL}
\ell|_M \equiv \inf_\F \ell = 0 \:.
\eeq
The parameter~$\s$ can be interpreted as the Lagrange parameter
corresponding to the volume constraint. For the derivation 
of the EL equations and further details we refer to~\cite[Section~2]{jet}.

\section{Causal Variational Principles on $\sigma$-Locally Compact Spaces} \label{seccvpsigma}
\subsection{Basic Definitions}\label{S Basic Definitions}
In the setup of causal variational principles in the non-compact setting (see~\S\ref{seccvp})
it is assumed that~$\F$ is a smooth manifold. Since this manifold structure is not needed
in what follows, we now slightly generalize the setting.

\begin{Def} Let~$\F$ be a second-countable, locally compact Hausdorff space,
and let the Lagrangian~$\L : \F \times \F \rightarrow \R^+_0$ be a symmetric and lower semi-continuous function
(see conditions~{\rm{(i)}} and~{\rm{(ii)}} in~\S\ref{seccvp}). Moreover, we assume that~$\L$
is strictly positive on the diagonal, i.e.\
\beq \label{strictpositive}
\L(x,x) > 0 \qquad \text{for all~$x \in \F$} \:.
\eeq
The {\bf{causal variational principle on $\sigma$-locally compact spaces}} is to minimize
the causal action~\eqref{Sact}
under variations of finite volume (see Definition~\ref{deffv}).
\end{Def} \noindent
Note that we do not impose the conditions~(iii) and~(iv) in~\S\ref{seccvp}.
For this reason, it is a-priori not clear whether the integrals in~\eqref{integrals} exist.
Therefore, we include this condition into our definition of a minimizer:
\begin{Def} \label{Definition minimizer} A regular Borel measure~$\rho$ is said to be a {\bf{minimizer}} of the causal action
under variations of finite volume if the difference~\eqref{integrals} is well-defined and
non-negative for all regular Borel measures~$\tilde{\rho}$ satisfying~\eqref{totvol},
\[ \big( \Sact(\tilde{\rho}) - \Sact(\rho) \big) \geq 0 \:. \]
\end{Def}
We point out that a minimizer again satisfies the EL equations~\eqref{EL}
(as is proved exactly as in~\cite[Lemma~2.3]{jet}).
The condition in~\eqref{strictpositive} is needed in order to avoid trivial minimizers supported
at a point where~$\L(x,x)=0$ (see~\cite[Section~1.2]{support}). Moreover, condition~\eqref{strictpositive} is
a reasonable assumption in view of~\cite[Exercise~1.2]{cfs}.

For clarity, we note that, following the conventions in~\cite{gardner+pfeffer},
by a Borel measure we mean a measure~$\rho : \B(\F) \to [0, + \infty]$ on the Borel $\sigma$-algebra~$\B(\F)$ which is {\em{locally finite}} (meaning that every point has an open neighborhood of finite volume).\footnote{We remark that a Borel measure is not usually taken to be locally finite by those working outside of \emph{topological} measure theory (cf.~\cite{gardner+pfeffer}).} 

In view of~\cite[Theorem 29.12]{bauer}, every Borel measure on~$\F$ is regular (in the sense that the measure of a set can be recovered by approximation with compact sets from inside and with open sets from outside). 
In particular, it is inner regular and therefore a Radon measure~\cite{schwartz-radon}.
More generally, every Borel measure on a Souslin space is regular by Meyer's theorem (see~\cite[Satz~VIII.1.17]{elstrodt}). 

A topological space which is locally compact and $\sigma$-compact is also referred to as
being {\em{$\sigma$-locally compact}} (see for example~\cite{steen+seebach}).
We note that every second-countable, locally compact Hausdorff space
is $\sigma$-compact (cf.~\cite[\S 29]{bauer}). 
%Moreover, every second countable space is separable (see e.g.~\cite[Proposition~4.5]{folland}).
Therefore,~$\F$ is a~$\sigma$-locally compact space. Moreover,
in view of \cite[Proposition 4.31]{folland} and \cite[Theorem~14.3]{willard}, the space~$\F$ is regular, and hence separable and metrizable by Urysohn's theorem (see for instance~\cite[Theorem 23.1]{willard}), where the resulting metric is complete (see~\cite[p.~185]{bauer}). 
Thus we can arrange that~$\F$ is a Polish space. Since each Polish space is Souslin, any Borel measure on $\F$ is regular, and therefore its support is given by~\eqref{suppdef}.

A metric space $X$ is said to have the \emph{Heine-Borel property} if every closed bounded subset is compact~\cite{williamson}.\footnote{In coarse geometry, such metric spaces are also called \emph{proper} (cf.~\cite[Definition 1.4]{roecoarse}). For instance, every connected complete Riemannian manifold is a proper metric space (see~\cite[Chapter 2]{roeindex}).} In this case, the corresponding metric is referred to as \emph{Heine-Borel metric}. Clearly, every Heine-Borel metric is complete.  According to~\cite[Theorem~2']{williamson}, every $\sigma$-locally compact Polish space is metrizable by a Heine-Borel metric. 
Since the topological space $\F$ is $\sigma$-locally compact and Polish we can arrange that bounded sets in $\F$
(with respect to the Heine-Borel metric) are relatively compact, i.e.\ have compact closure.

Moreover, in order to construct solutions of the EL equations, we first impose the following
assumption (see Section~\ref{secmincr}).
\begin{Def} \label{defcompactrange}
	The Lagrangian has {\bf{compact range}} if for every compact set~$K \subset \F$ there is
	a compact set~$K' \subset \F$ such that
	\[ \L(x,y) = 0 \qquad \text{for all~$x \in K$ and~$y \not \in K'$}\:. \]
\end{Def} \noindent
Later on we will show that this assumption can be weakened (see Section~\ref{Chapter 5}). 

\subsection{Existence of Minimizers on Compact Subsets}\label{S Existence compact}
Our strategy is to exhaust~$\F$ by compact sets, to minimize on each compact set,
and to analyze the limit of the resulting measures. In preparation, we now consider
the variational principle on a compact subset~$K \subset \F$.
Since the restriction of a Borel measure (according to \cite[Definition~2.1]{gardner+pfeffer}) to~$K$ has finite volume, by rescaling we
may arrange that the total volume equals one. This leads us to the variational principle
\[ \text{minimize} \qquad \Sact_K(\rho) := \int_K d\rho(x) \int_K d\rho(y)\: \L(x,y) \]
in the class
\[ 
\rho \;\in\; \M_K := \{ \text{normalized Borel measures on~$K$} \}\:, \]
where {\em{normalized}} means that~$\rho(K)=1$ (since we are not concerned with probability theory,
in our context it seems preferable to avoid the notion of a probability measure). 

Existence of minimizers follows from abstract compactness arguments
in the spirit of~\cite[Section~1.2]{continuum}. We give the proof in detail
because the generalization to the lower semi-continuous setting is not quite obvious.
\begin{Thm} \label{thm1}
Let~$K \subset \F$ be compact. Moreover, let~$(\rho_k)_{k \in\N}$ be a minimizing sequence 
in~$\M_K$ for the action $\Sact_K$, i.e.\
	\begin{align*}
	\lim_{k \to \infty} \Sact_K\left(\rho_k \right) = \inf_{\text{$\rho \in \M_K$}} \Sact_K(\rho) \:. 
\end{align*}
Then the sequence $(\rho_k)_{k \in \N}$ contains a subsequence which converges weakly  to a minimizer~$\rho_K \in \mathfrak{M}_K$. 
\end{Thm}
\Proof Let $(\rho_k)_{k \in \N}$ be a minimizing sequence. For clarity, note that the compact subset~$K \subset \F$ is a locally compact Hausdorff space. Moreover, the continuous, real-valued functions on~$K$,
denoted by~$C(K)$, form a normed vector space (with respect to the sup norm~$\|\cdot\|_{\infty}$), and the functions in~$C(K)$ are all bounded and have compact support, i.e.\
$C(K) = C_b(K) = C_c(K)$. For each~$k \in \N$, the mapping
\begin{align*}
I_k \,:\, C(K) \to \R, \qquad I_k(f) := \int_{K} f(x) \:d{\rho}_k(x)
\end{align*}
defines a continuous positive linear functional. Since
	\begin{align*}
	\|I_k\| := \sup_{f\in C(K)\atop \|f\|\le 1} \left| \int_{K} f(x) \:d{\rho}_k(x) \right| \le \|\rho_k\|(K)
	\end{align*}
	and $\|\rho_k\|(K)= \rho_k(K) = 1$ for all $k\in \N$ (where $\|\cdot\|(K)$ denotes the total variation, and $\|\cdot\|$ the operator norm on~$C(K)^{\ast}$), the sequence $(I_k)_{k \in \N}$ is bounded in $C(K)^{\ast}$.
In view of the Banach-Alaoglu theorem, a subsequence~$(I_{k_j})_{j \in \N}$ converges
to a linear functional~$I \in C(K)^{\ast}$ in the weak*-topology,
\[ I_{k_j} \rightharpoonup^* I \in C(K)^{\ast} \:. \]
Applying the Riesz representation theorem, we obtain a regular Borel measure $\rho_K$ such that
\[ I(f) = \int_{K} f(x) \:d{\rho_K}(x) \qquad \text{for all~$f \in C(K)$}\:. \]
Since~$\rho_K(K) = I(1_K) = \lim_{j \rightarrow \infty} I_{k_j}(1_K)= 1$
(where~$1_K$ is the function which is identically equal to one), one sees that~$\rho_K$ is again normalized.

It remains to show that ${\rho_K}$ is a minimizer. Since $K$ is compact, $\sigma$-compactness of~$K$ implies that the measure space $(K,\B(K))$ is $\sigma$-finite (according to \cite[\S 7]{halmosmt}; this also results from the fact that any Borel measure is locally finite and $K$ is second-countable). Due to \cite[\S 35, Theorem B]{halmosmt}, for all $j \in \N$ there is a uniquely determined product measure 
$$\eta_{k_j} := {\rho}_{k_j} \times {\rho}_{k_j} \colon \B(K) \otimes \B(K) \to \R$$ 
(see also~\cite[Theorem 7.20]{folland}) such that
\begin{align*}
{\eta}_{k_j}(A \times B) := {\rho}_{k_j}(A) \cdot {\rho}_{k_j}(B) \:,
\end{align*}
where~$A\times B \in \B(K) \otimes \B(K)$.
Since~$K \subset \F$ is a second-countable Hausdorff space, it is separable according to~\cite[\S 5F]{willard}, and the Cartesian product~$K \times K$ is compact (see e.g.~\cite[Theorem~3.2.3]{engelking}). Moreover, any countable product of second-countable topological spaces is again second-countable and thus separable. By~\cite[Theorem 2.8]{billingsley} we obtain
weak convergence 
\[ \eta_{k_j} = \rho_{k_j} \times \rho_{k_j} \rightharpoonup \rho_K \times \rho_K =: \eta_K \:. \]
In particular, $(\eta_{k_j})_{j \in \N}$ is a sequence of normalized Borel measures, and $\eta_K$ is a normalized Borel 
measure on $K \times K$. Since $K \times K$ is metrizable due to~\cite[\S 34]{munkres}, and the Lagrangian~$\L|_{K \times K} : K \times K \to \R_0^+$ is a measurable non-negative real valued function on $K \times K$, Fatou's lemma for sequences of measures~\cite[eq.~(1.5)]{feinberg} yields
	\begin{align*}
	\Sact_K({\rho_K}) &= \int_{K} \int_{K} \L(x,y) \:d{\rho_K}(x) \:d{\rho_K}(y) = \iint_{K \times K} \L(x,y) \:d{\eta_K}(x,y) \\
	&\le \liminf_{j \to \infty} \iint_{K \times K} \L(x,y)\: d{\eta}_{k_j}(x,y) = \liminf_{j \to \infty}  \int_{K} \int_{K} \L(x,y) \:d{\rho}_{k_j}(x) \:d{\rho}_{k_j}(y) \\
	&= \liminf_{j \to \infty} \Sact_K(\rho_{k_j}) \le \lim_{j \to \infty} \Sact_K(\rho_{k_j}) =   \inf_{\text{$\rho \in \M_K$}} \Sact_K(\rho) \:.\phantom{\int}
	\end{align*}
Hence $\rho_K$ is a minimizer of the action $\Sact_K$.
\QED

A minimizing measure~$\rho_K \in \mathfrak{M}_K$ satisfies the corresponding Euler-Lagrange equations, which in analogy to~\eqref{EL} read
\beq \label{ELK}
\ell_K|_{\supp \rho_K} \equiv \inf_K \ell_K = 0 \:,
\eeq
where~$\ell_K : \F \rightarrow \R$ is the function
\beq \label{lKdef}
\ell_K(x) := \int_K \L(x,y)\: d\rho_K(y) - \s\:,
\eeq
and~$\s>0$ is a suitably chosen parameter. For clarity, we point out that the integral in~\eqref{lKdef} is strictly positive for the following reason: Each normalized measure~$\rho \in \mathfrak{M}_K$ has non-empty support; thus for any~$x \in \supp \rho \not= \varnothing$, the fact that the Lagrangian is lower semi-continuous as well as strictly positive on the diagonal~\eqref{strictpositive} implies that there is an open neighborhood~$U$ of~$x$ such that~$\L(x,y) > \L(x,x)/2 > 0$ for all~$y \in U$ and~$\rho(U) > 0$ in view of~\eqref{suppdef}. As a consequence, 
\[ \int_{K} \L(x,y) \: d\rho(y) \ge \int_U \L(x,y) \: d\rho(y) > 0 \:. \]
Comparing this inequality for~$x \in K$ with~\eqref{ELK} and~\eqref{lKdef}, we conclude that
the last integral is strictly positive and constant on~$K$.
%
%Using that this inequality holds for every $x \in \supp \rho$
%and that the integral is lower semi-continuous in~$x$,
%we conclude that
%\[ \inf_{x \in K} \int_{K} \L(x,y) \: d\rho(y) > 0 \:. \]
%Choosing any~$x 
%
%From~\eqref{ELK} we deduce that~$\inf_{x \in \F} \ell_K(x) > 0$. 

\section{Minimizers for Lagrangians of Compact Range} \label{secmincr}

\subsection{Construction of a Global Borel Measure} \label{secconstglobal}
Let $(K_n)_{n \in \N}$ be an exhaustion of the $\sigma$-locally compact space~$\F$ by compact sets such that each compact set is contained in the interior of its successor (see e.g.~\cite[Lemma~29.8]{bauer}).
For every~$n \in \N$, we let~$\rho_{K_n}$ in~$\M_{K_n}$ be a corresponding minimizer on~$K_n$
as constructed in Theorem~\ref{thm1}.
We extend these measures by zero to~$\B(\F)$,
	\begin{align}\label{rhon}
\rho^{[n]}(A) := 
	\lambda_n \:\rho_{K_n}(A \cap K_n) \:,
\end{align}
where~$\lambda_n$ are positive parameters which will be chosen
such that the parameter~$\s$ in the EL equations~\eqref{ELK} and~\eqref{lKdef}
is equal to one. Thus
\beq \label{ELn}
\ell^{[n]}\big|_{\supp \rho^{[n]}} \equiv \inf_{K_n} \ell^{[n]} = 0 \:,
\eeq
where
\beq \label{lndef}
\ell^{[n]}(x) := \int_\F \L(x,y)\: d\rho^{[n]}(y) - 1 \:.
\eeq
For clarity, we point out that the measures~$\rho^{[n]}$ are {\em{not normalized}}. More precisely,
\[ \rho^{[n]}(\F) = \lambda_n \:, \]
and the sequence~$(\lambda_n)_{n \in \N}$ will typically be unbounded.

\begin{Lemma} \label{lemmaupper} For every compact subset~$K \subset \F$ there is a constant~$C_K>0$ such that
\[ \rho^{[n]}(K) \leq C_K \qquad \text{for all~$n \in \N$}\:. \]
\end{Lemma}
\Proof Since~$\L(x,.)$ is lower semi-continuous and strictly positive at~$x$
(see~\eqref{strictpositive}), there is an open neighborhood~$U(x)$ of~$x$ with
\[ \L(y,z) \geq \frac{\L(x,x)}{2} > 0 \qquad \text{for all~$y,z \in U(x)$}\:. \]
Covering~$K$ by a finite number of such neighborhoods~$U(x_1), \ldots, U(x_L)$,
it suffices to show the inequality for the sets~$K \cap U(x_\ell)$ for any~$\ell \in \{1,\ldots, L\}$.
Moreover, we choose~$N$ so large that~$K_N \supset K$ and fix~$n \geq N$.
If~$K \cap \supp \rho^{[n]} = \varnothing$, there is nothing to prove.
Otherwise, there is a point~$z \in K \cap \supp \rho^{[n]}$. Using the EL equations~\eqref{ELn} at~$z$,
it follows that
\[ 1 = \int_\F \L(z,y)\: d\rho^{[n]}(y)
\geq \int_{U(x_\ell)} \L(z,y)\: d\rho^{[n]}(y) \geq 
\frac{\L(x_\ell,x_\ell)}{2}\: \rho^{[n]}(U(x_\ell)) \:. \]
Hence
\begin{align}\label{(lemmaupper)}
	 \rho^{[n]}(U(x_\ell)) \leq \frac{2}{\L(x_\ell,x_\ell)}\:. 
\end{align}
This inequality holds for any~$n \geq N$.
Let~$c(x_\ell)$ be the maximum of~$2/\L(x_\ell,x_\ell)$ and~$\rho^{[1]}(U(x_\ell)), \ldots, \rho^{[N-1]}(U(x_\ell))$. 
Since the open sets~$U(x_1), \ldots, U(x_L)$ cover~$K$,~we finally introduce the constant~$C_K$ as the sum of the constants~$c(x_1), \ldots, c(x_L)$. 
\QED

Now we proceed as follows. Denoting by $(K_n)_{n \in \N}$ the above exhaustion of~$\F$ by compact sets, we first restrict the measures~$\rho^{[n]}$ to the compact set~$K_1$.
According to Lemma~\ref{lemmaupper}, the resulting sequence of measures is bounded. Therefore,
a subsequence converges as a measure on~$K_1$ (using again the Banach-Alaoglu theorem 
and the Riesz representation theorem). Out of the resulting subsequence $(\rho^{[1, n_k]})_{k \in \N}$, we then choose a subsequence
of measures $(\rho^{[2, n_k]})_{k \in \N}$ which converges weakly on~$K_2$. We proceed iteratively and denote the resulting
diagonal sequence by
\beq \label{rhorund}
\rho^{(k)} := \rho^{[k, n_k]} \qquad \text{for all~$k \in \N$}\:.
\eeq
\noindent
In the following, we restrict attention to the compact exhaustion $(K_m)_{m \in \N}$, where for convenience 
by $K_m$ we denote the sets $K_{n_m}$ for~$m \in \N$ (thus $\rho^{(m)}$ is a minimizer on~$K_m$ for each $m \in \N$). 
By construction, the sequence $(\rho^{(k)}|_{K_m})_{k \in \N}$ converges weakly to some measure $\rho_{K_m}$ for every $m \in \N$, i.e.
\begin{align}\label{(weak convergence on Km)}
\rho^{(k)}|_{K_m} \rightharpoonup \rho_{K_m} \qquad \text{for all $m \in \N$} \:.
\end{align}
In particular, 
\begin{align*}%\label{(weak convergence rho tilde)}
\lim_{k \to \infty} \rho^{(k)}(K_m) = \rho_{K_m}(K_m) \qquad \text{for all $m \in \N$} \:. 
\end{align*}
Denoting the interior of $K_m$ by $K_m^{\circ}$, we claim that for all $n \ge m$, 
\begin{align}\label{(coincidence)}
\int_{K_m} f \: d\rho_{K_m} = \int_{K_n} f \: d\rho_{K_n} \qquad \text{for all $f \in C_c(K_m^{\circ})$} \:. 
\end{align}
Namely, in view of $C_c(K_m^{\circ}) \subset C_b(K_m) \cap C_b(K_n)$, weak convergence \eqref{(weak convergence on Km)} yields 
\begin{align*}
\int_{K_m} f \: d\rho_{K_m} = \lim_{k \to \infty} \int_{K_m} f \: d\rho^{(k)}|_{K_m} \stackrel{\text{($\star$)}}{=} \lim_{k \to \infty} \int_{K_n} f \: d\rho^{(k)}|_{K_n} = \int_{K_n} f \: d\rho_{K_n}
\end{align*}
for all $f \in C_c(K_m^{\circ})$, 
where in ($\star$) we made use of the fact that $\supp f \subset K_m^{\circ} \subset K_n$. 

In order to construct a global measure $\rho$ on $\F$, we proceed by introducing the functional
\begin{align*}
I : C_c(\F) \to \R \:, \qquad f \mapsto I(f) := \lim_{n \to \infty} \int_{\F} f(x) \: d\rho_{K_n}(x) \:.
\end{align*}
Note that, in view of \eqref{(coincidence)}, the last integral is independent of $n \in \N$ for sufficiently large integers~$n \in \N$, and thus well-defined. Thus $I : C_c(\F) \to \R$ defines a positive functional. Making use of the fact that $\F$ is locally compact, the Riesz representation theorem \cite[Darstellungssatz~VIII.2.5]{elstrodt} or \cite[Theorem 29.1 and Theorem 29.6]{bauer} yields the existence of a uniquely determined Radon measure $\rho : \B(\F) \to [0, \infty]$ (i.e.\ an inner regular, locally finite measure \cite{schwartz-radon}) such that
\begin{align}\label{rhoglobal}
I(f) = \int_{\F} f \: d\rho \qquad \text{for all $f \in C_c(\F)$}
\end{align}
(also see~\cite[Definition 4 and Theorem 5]{edwards}). 
From the equality
\begin{align}\label{(vague convergence rhoKm)}
\lim_{n \to \infty} \int_{\F} f(x) \: d\rho_{K_n}(x) = I(f) = \int_{\F} f \: d\rho \qquad \text{for all $f \in C_c(\F)$} 
\end{align}
we conclude that the sequence $(\rho_{K_n})_{n \in \N}$ converges vaguely to the measure~$\rho$ (for details see \cite[Definition 30.1]{bauer}). Moreover, given $f \in C_c(\F)$, we know that $\supp f \subset K_m^{\circ}$ for sufficiently large $m \in \N$. Hence from vague convergence \eqref{(vague convergence rhoKm)} and weak convergence~\eqref{(weak convergence on Km)} we conclude that for any $\varepsilon > 0$ there exists $m' \in \N$ such that for all~$m \ge m'$, 
\begin{align*}
\left|\int_{\F} f \: d\rho - \int_{\F} f \: d\rho^{(k)} \right| \le \left|\int_{\F} f \: d\rho - \int_{\F} f \: d\rho_{K_m} \right| + \left|\int_{\F} f \: d\rho_{K_m} - \int_{\F} f \: d\rho^{(k)}|_{K_m} \right| < \varepsilon \:. 
\end{align*}
Since $f \in C_c(\F)$ was arbitrary, we also conclude that 
\begin{align}\label{(vague convergence)}
\rho^{(k)} \stackrel{v}{\to} \rho \qquad \text{vaguely} \:. 
\end{align}
Regularity follows from Ulam's theorem \cite[Satz~VIII.1.16]{elstrodt} and the fact that $\F$ is Polish. 
In Appendix~\ref{Appendix nontrivial} (see Lemma~\ref{Lemma nontrivial})
it is shown that the measure~$\rho$ given by~\eqref{rhoglobal} is non-zero
and possibly has infinite total volume.

Similar to~\eqref{lndef}, we introduce the notation
\beq \label{lndef'}
\ell^{(n)}(x) := \int_\F \L(x,y)\: d\rho^{(n)}(y) - 1 \:.
\eeq
In particular, the following EL equations hold,
\beq \label{ELn'}
\ell^{(n)}\big|_{\supp \rho^{(n)}} \equiv \inf_{K_n} \ell^{(n)} = 0 \:.
\eeq

\subsection{Derivation of the Euler-Lagrange Equations}\label{Section Derivation}
In this section, we assume that the Lagrangian~$\L$ is continuous and of compact range (see Definition~\ref{defcompactrange}). Our goal is 
to prove the following result.
\begin{Thm}[\textbf{Euler-Lagrange equations}] \label{Theorem EL equations locally compact}
Assume that~$\L$ is continuous and of compact range.
Then the measure~$\rho$ constructed in~\eqref{rhoglobal} satisfies the Euler-Lagrange equations
\begin{align}\label{(EL equations)}
\ell|_{\supp \rho} \equiv \inf_{\text{$x \in \F$}} \ell(x) = 0 \:,
\end{align}
where~$\ell \in C(\F)$ is defined by
	\begin{align}\label{(Lagrange functional)}
	\ell (x) := \int_{\F} \L(x,y) \:d\rho(y) - 1 \:.
	\end{align}
\end{Thm}
In order to prove Theorem~\ref{Theorem EL equations locally compact} we proceed in several steps.
The proof will be completed at the end of this subsection. 

\begin{Lemma} \label{Lemma support}
	Assume that the sequence of measures $(\rho^{(n)})_{n \in \N}$ converges vaguely to a regular Borel measure~$\rho \not= 0$. 
	Then 
	for every~$x \in \supp \rho$ there is a sequence~$(x_k)_{k \in \N}$ and a subsequence~$\rho^{(n_k)}$
	such that~$x_k \in \supp \rho^{(n_k)}$
	for all~$k \in \N$ and~$x_k \rightarrow x$.
\end{Lemma}
\Proof Assume conversely that for some $x \in \supp \rho$ there is no such subsequence. Then there is an 
open neighborhood~$U$ of~$x$ which does not intersect the support of the measures~$\rho^{(n)}$ for almost all~$n \in \N$.
As a consequence, for every compact neighborhood~$V$ of~$x$ with $V \subset U$ (which exists by \cite[Proposition 4.30]{folland}) we have
\[ \rho^{(n)}(V) = 0 \qquad \text{for almost all~$n \in \N$} \:, \]
and
\begin{align*}
\rho(V) > 0
\end{align*}
in view of \eqref{suppdef}. 
Since $U$ is an open cover of the compact Hausdorff space $V$, there exists $f \in C_c(U; [0,1])$ with $f|_V \equiv 1$ (see for instance~\cite[Lemma 2.92]{aliprantis}). Thus vague convergence
$\rho^{(n)} \stackrel{v}{\to} \rho$ yields the contradiction
\begin{align*}
0 < \int_{\F} f \: d\rho = \lim_{n \to \infty} \int_{\F} f \: d\rho^{(n)} = 0 \:,
\end{align*}
which proves the claim.
%Taking the limit~$n \rightarrow \infty$ and using~\eqref{(weak convergence)}, we obtain~$\rho(V) = 0$.
%This is a contradiction to the assumption that~$x \in \supp \rho$.
\QED

For notational simplicity, we denote the subsequence~$\rho^{(n_k)}$ again by~$\rho^{(k)}$.

\begin{Lemma}\label{Lemma Continuity of ell}
Let $\L$ be continuous and of compact range. Then the function~$\ell \colon \F \to \R$ defined by \eqref{(Lagrange functional)} is continuous.
\end{Lemma}
\begin{proof}
	For any $x \in \F$, let $(x_n)_{n \in \N}$ be an arbitrary sequence in $\F$ converging to $x$, and let $U$ be an open, relatively compact neighborhood of $x$ (which exists by local compactness of $\F$). Since $\L$ is of compact range, there is a compact set $K' \subset \F$ such that $\L(\tilde{x}, y) = 0$ for all $\tilde{x} \in K := \overline{U}$
and $y \notin K'$. Since the sequence~$(x_n)_{n \in \N}$ converges to $x$, there is an integer $N \in \N$ such that $x_n\in U$ for all $n \ge N$. By continuity of $\L$, the mapping~$\L : K \times K' \to \R$ is bounded. Therefore, the functions $\L(x_n, \cdot) : K' \to \R$ are uniformly bounded for all $n \ge N$. Thus Lebesgue's dominated convergence theorem yields 
	\begin{align*}
		\ell(x) &= \int_{K'} \L(x,y) \:d\rho(y) - 1 = \int_{K'} \lim_{n \to \infty} \L(x_n,y) \:d\rho(y) -1 \\
		&= \lim_{n \to \infty} \int_{K'} \L(x_n,y) \:d\rho(y) - 1 = \lim_{n \to \infty} \ell(x_n) \:, 
	\end{align*}
	 proving continuity of $\ell$. 
\end{proof}

In the next proposition, we show that the sequence $(\ell^{(n)})_{n \in \N}$ converges pointwise to~$\ell$.
Choosing~$K=\{x\}$ in Definition~\ref{defcompactrange}, we denote the corresponding compact set~$K'$ by~$K_x$, i.e.
\beq \label{Kxdef}
\L(x,y) = 0 \qquad \text{for all~$y \not \in K_x$} \:.
\eeq
\begin{Prp}\label{Proposition Pointwise Convergence Locally Compact}
Let $\L$ be continuous and of compact range, and let~$(\ell^{(n)})_{n \in \N}$ and~$\ell$ be the functions defined in~\eqref{lndef'} and~\eqref{(Lagrange functional)}, respectively. 
Then $(\ell^{(n)})_{n \in \N}$ converges pointwise to $\ell$, i.e.
	\begin{align}\label{(pointwise convergence)}
	\lim_{n \to \infty} \ell^{(n)} (x) = \ell (x) \qquad \text{for all $x \in \F$} \:.
	\end{align}
\end{Prp}
\begin{proof}
	Let $x \in \F$. Since $\L$ is assumed to be of compact range, using the notation~\eqref{Kxdef}
	implies that 
	$\L(x,y) = 0$ for all $y \notin K_x$. 
	We conclude that $\L(x, \cdot) \in C_c(\F)$, and thus by vague convergence \eqref{(vague convergence)} we obtain 
	\begin{align*}
	\ell(x) = \int_{\F} \L(x,y) \:d\rho(y) - 1 
	= \lim_{n \to \infty} \int_{\F} \L(x,y) \:d\rho^{(n)}(y) - 1 = \lim_{n \to \infty} \ell^{(n)}(x) \:.
	\end{align*}
	Since $x \in \F$ is arbitrary, the sequence $(\ell^{(n)})_{n \in \N}$ converges pointwise to $\ell$. 
\end{proof}

Our proof of Theorem \ref{Theorem EL equations locally compact} will be based on equicontinuity of the family~$(\ell^{(n)}|_K)_{n \in \N}$ for arbitrary compact subsets $K \subset \F$.
We know that the functions~$\ell^{(n)}$ are continuous and uniformly bounded on compact sets.
However, as can be seen from the example~$(f_n)_{n \in \N}$ with
\[ f_n : [0,1] \rightarrow \R, \qquad f_n(x) = \sin nx \qquad \text{for all~$n \in \N$}\:, \]
these conditions are in general not sufficient to ensure equicontinuity.
Nonetheless, the additional assumption that the Lagrangian $\L \colon \F \times \F \to \R_0^+$ is of compact range
(see Definition~\ref{defcompactrange}) gives rise to equicontinuity
of the family~$(\ell^{(n)}|_K)_{n \in \N}$, as the following proposition shows.

\begin{Prp}\label{Proposition Equicontinuity Locally Compact}
	Let $\L$ be continuous and of compact range. Then for any compact subset $K \subset \F$, the family $F_K := \{\ell^{(n)}|_K \colon n \in \N \}$ is equicontinuous.
\end{Prp}
\begin{proof}
	Consider an arbitrary compact set $K \subset \F$. In order to prove equicontinuity of~$F_K$, we have to show that for every $\varepsilon >0$ and every $x \in K$ there is a corresponding neighborhood $V = V(x)$ of $x$ with 
	\begin{align*}
	\sup_{f \in F_K} \sup_{z \in V} \left|f(x) - f(z)\right| < \varepsilon\: . 
	\end{align*}
	Let $x \in K$ and consider an arbitrary $\varepsilon > 0$. Since $\L$ is of compact range, there is a compact set $K' \subset \F$ such that
\beq \label{LKKp}
\L(\tilde{x},y) = 0 \qquad \text{for all $\tilde{x} \in K$ and~$y \notin K'$} \:.
\eeq
In view of Lemma~\ref{lemmaupper} there is a positive constant~$C_{K'} > 0$ such that~$\rho^{(n)}(K') \le C_{K'}$
for all~$n \in \N$. 
	
	Since $\L$ is continuous and $K \times K'$ is compact, the mapping  
	\[\L|_{K \times K'} \colon K \times K' \to \R \]
	is uniformly continuous. Moreover, in view of~\eqref{LKKp}, the same is true for~$\L|_{K \times \F}$.
 Hence for every $\varepsilon > 0$ there is a $\delta >0$ such that 
\[ \Big| \L|_{K \times \F}(x_1, y_1) - \L|_{K \times \F}(x_2, y_2) \Big| < \varepsilon \qquad \text{for all $(x_2, y_2) \in B_{\delta}((x_1, y_1))$} \:. \]
	Choosing $\delta > 0$ such that $\left|\L|_{K \times \F}(x,\cdot) - \L|_{K \times \F}(z, \cdot) \right| < \varepsilon/C_{K'}$ for all $z \in B_{\delta}(x) \cap K$, we obtain
	\begin{align*}
	\sup_{n \in \N} &\sup_{z \in B_{\delta}(x) \cap K} \left|\ell^{(n)}|_K(x) - \ell^{(n)}|_K(z)\right| \\
	&= \sup_{n \in \N} \sup_{z \in B_{\delta}(x) \cap K} \left|\int_{\F} \Big( \L|_{K \times \F}(x,y) - \L|_{K \times \F}(z,y) 
	\Big)\: d\rho^{(n)}(y) \right| \\
	&\le \sup_{n \in \N} \sup_{z \in B_{\delta}(x) \cap K} \int_{K'} \Big| \L|_{K \times \F}(x,y) - \L|_{K \times \F}(z,y) \Big|\:d\rho^{(n)}(y) \\
	&< \sup_{n \in \N} \rho^{(n)}(K') \, \frac{\varepsilon}{C_{K'}} \le \varepsilon \:.
	\end{align*}
	This yields equicontinuity of $F_K$ as desired.
\end{proof}

After these preparations, we are able to prove Theorem \ref{Theorem EL equations locally compact}.
\begin{proof}[Proof of Theorem \ref{Theorem EL equations locally compact}]
	Let $(K_n)_{n \in \N}$ be a compact exhaustion of $\F$, and let $(\rho^{(n)})_{n \in \N}$ be the corresponding sequence of vaguely converging measures according to~\eqref{rhorund} such that~\eqref{lndef'} and~\eqref{ELn'} hold. The main idea of the proof is to make use of pointwise convergence \eqref{(pointwise convergence)} and equicontinuity of the sequence $(\ell^{(n)}|_K)_{n \in \N}$ for arbitrary compact sets $K \subset \F$ as established in Proposition~\ref{Proposition Pointwise Convergence Locally Compact} and Proposition~\ref{Proposition Equicontinuity Locally Compact}, respectively. 
	
	First of all, application of Proposition \ref{Proposition Pointwise Convergence Locally Compact} shows that $\ell(x) \ge 0$ for every $x \in \F$. Namely, since $\rho^{(n)}$ is a minimizer of the action $\Sact_{K_n}$ for every~$n \in \N$, and~$x$ is contained in all compact sets~$(K_{n})_{n \ge N}$ for some integer $N = N(x)$, we have
	\begin{align}
	\ell(x) \stackrel{\eqref{(pointwise convergence)}}{=} \lim_{n \to \infty} \ell^{(n)}(x) = \lim_{n \to \infty} \ell^{(n)}|_{K_{n}}(x) \stackrel{\eqref{ELn'}}{\ge} 0 \qquad \text{for all~$x \in \F$} \:. \label{(bigger 0)}
	\end{align}
	In order to derive the EL equations~\eqref{(EL equations)}, 
	it remains to prove that $\ell(x)$ vanishes for every $x \in\supp \rho$. Since $\F$ is locally compact, each $x \in \supp \rho$ is contained in a compact neighborhood $K_x$ (cf.~\eqref{Kxdef}) such that $\L(x, .)$ vanishes outside $K_x$. Due to vague convergence~$\rho^{(n)} \to \rho$ as~$n \to \infty$ (cf.~\eqref{(vague convergence)}), by virtue of Lemma \ref{Lemma support} there exists a sequence~$x^{(n)} \to x$ as $n \to \infty$ such that~$x^{(n)} \in \supp \rho^{(n)}$ for every $n \in \N$. We choose~$N' \in \N$ such that~$x^{(n)} \in K_x$ for all~$n \geq N'$. For this reason, it suffices to focus on the restriction~$\ell|_{K_x}$. 
	Equicontinuity of the family~$\{\ell^{(n)}|_{K_x} \colon n \in \N \}$ (see Proposition~\ref{Proposition Equicontinuity Locally Compact}) yields
	\[ %\label{(equicontinuity)}
	\lim_{n \to \infty} \sup_{k \in \N} \left|\ell^{(k)}|_{K_x}\big( x^{(n)} \big) -\ell^{(k)}|_{K_x}(x) \right| =0 \:. \]
	Moreover, the expression
	\[ %\label{(pointwise convergence applied)}
	\lim_{n \to \infty} \left|\ell^{(n)}|_{K_x}(x) - \ell|_{K_x}(x)\right| = 0 \qquad \text{for all~$x \in K_x$} \]
	holds in view of pointwise convergence \eqref{(pointwise convergence)}. Taken together, for every~$x \in \supp \rho$ we finally obtain
	\begin{align*}
	&\lim_{n \to \infty}\left|\ell^{(n)} \big(x^{(n)} \big) - \ell(x) \right| =  \lim_{n \to \infty} \left|\ell^{(n)}|_{K_x} \big(x^{(n)} \big) - \ell|_{K_x}(x) \right| \\
	&\qquad \le \lim_{n\to \infty} \left|\ell^{(n)}|_{K_x} \big(x^{(n)} \big) -\ell^{(n)}|_{K_x}(x) \right| + \lim_{n \to \infty} \left|\ell^{(n)}|_{K_x}(x) - \ell|_{K_x}(x)\right| = 0 \:.
	\end{align*}
	In view of \eqref{(bigger 0)}, the Euler-Lagrange equations \eqref{(EL equations)} hold due to
	\begin{align*}
	\ell(x) = \lim_{n \to \infty} \ell^{(n)} \big(x^{(n)} \big) \stackrel{\eqref{ELn'}}{=} 0 \qquad \text{for all~$x \in \supp \rho$} \:,
	\end{align*}
	which completes the proof.
\end{proof}

In the remainder of this subsection, we discuss the properties (iii) and (iv) in~\S\ref{seccvp}.
Condition~(iii) holds by construction because we are working with locally finite measures
(see~\S\ref{S Basic Definitions}). Condition~(iv) does not hold in general,
but it can be checked a-posteriori for a constructed measure~$\rho$.
Under suitable assumptions on~$\L$, however, 
this condition can even be verified a-priori, i.e.\ without knowing~$\rho$.
This is exemplified in the following lemma.

\begin{Lemma}\label{Lemma conditions}
	Let $\L : \F \times \F \to \R$ be continuous and of compact range. Moreover, assume that the following conditions hold:
	\begin{enumerate}[leftmargin=2em]
		\item[\rm{(a)}] $c := \inf_{x \in \F} \L(x,x) > 0$. \\[-1em]
		\item[\rm{(b)}] $\sup_{x,y \in \F} \L(x,y) \le \mathscr{C} < \infty$.  \\[-0.9em]
		\item[\rm{(c)}] There is an integer~$N>0$ such that every~$K_x$ (as defined in~\eqref{Kxdef}) can be covered by
		open sets~$U_1, \ldots, U_N$ with the property that for all~$i \in \{1,\ldots, N\}$,
		\[ \L(x,y) > \frac{c}{2} \qquad \text{for all~$y \in U_i$} \:. \]
	\end{enumerate}
	Then the measure~$\rho$ constructed in~\eqref{rhoglobal} satisfies condition~{\rm{(iv)}} in~\S\ref{seccvp}.
\end{Lemma}
\begin{proof} Since~$\L$ is continuous and of compact range, we have
	\begin{align*}
	\int_{\F} \L(x,y) \:d\rho(y) = \int_{K_x} \L(x,y) \:d\rho(y) \le \sup_{y \in K_x} \L(x,y) \:\rho(K_x) < \infty \:,
	\end{align*}
	showing that~$\L(x, \cdot)$ is $\rho$-integrable for every $x \in \F$. It remains to prove that
	\[\sup_{x \in \F} \int_{\F} \L(x,y) \:d\rho(y) < \infty \:. \]
	Since $\L : \F \times \F \to \R$ satisfies (a)--(c), inequality \eqref{(lemmaupper)} yields 
	\begin{align*}
	\rho(U_i) \le \sup_{x \in \F} \frac{2}{\L(x,x)} \le \frac{2}{c} \qquad \text{for all~$i\in \{1,\ldots, N\}$} \:.
	\end{align*}
	Thus we obtain
	\begin{align*}
	\sup_{x \in \F} \int_{\F} \L(x,y) \:d\rho(y) \le \sup_{x,y \in \F} \L(x,y) \:\rho(K_x) \le \mathscr{C}
	\sum_{i=1}^N \rho(U_i) < \frac{2 \mathscr{C} N}{c} 
	< \infty 
	\end{align*}
	as desired. 
\end{proof}

\subsection{Existence of Minimizers under Variations of Compact Support}\label{subsection compact variation}
Apart from technical convenience, it
is natural and sufficient for many applications to restrict attention to
variations of compact support (in contrast to the more general variations of finite volume
as introduced in Definition~\ref{deffv} and Definition~\ref{Definition minimizer}).
\begin{Def} \label{Def global}
	A measure $\rho \in \mathfrak{B}_{\F}$ is said to be a {\bf{minimizer under variations of compact support}}
	of the causal action if for any~$\tilde{\rho} \in \mathfrak{B}_{\F}$ which satisfies~\eqref{totvol} such that the signed measure~$(\rho-\tilde{\rho})$ is compactly supported, the inequality 
	\[ \big( \Sact(\tilde{\rho}) - \Sact(\rho) \big) \geq 0 \]
	holds.
\end{Def}
The goal of this subsection is to prove that the measure~$\rho$ constructed in~\eqref{rhoglobal}
is a minimizer under variations of compact support.
Before stating our result (see Theorem~\ref{Lemma compact minimizer} below),
we show that the difference \eqref{integrals} is well-defined.
Indeed, considering variations of compact support,
the signed measure~$\mu := \tilde{\rho} - \rho$ is compactly supported. Considering its Jordan decomposition~$\mu = \mu^+ - \mu^-$ (see e.g.~\cite[\S 29]{halmosmt}), the measures~$\mu^+$ and~$\mu^-$ have compact support. Hence, using that the Lagrangian is continuous,
\begin{align*}
	\int_{\F} d\mu^+ (x) \:\ell(x) \le \left(\sup_{x \in \supp \mu^+} \ell(x)\right) \mu^+(\supp \mu^+) < \infty \:,
\end{align*}
and similarly for $\mu^-$. Now we can proceed as in the proof of \cite[Lemma 2.1]{jet}
to conclude that all the integrals in~\eqref{integrals} are well-defined and finite.

\begin{Thm}\label{Lemma compact minimizer}
	Assume that the Lagrangian $\L$ is continuous and of compact range. Then the measure~$\rho$ constructed in \eqref{rhoglobal} is a minimizer under variations of compact support.
\end{Thm}
\begin{proof}
	Let $\tilde{\rho}$ be a regular Borel measure on $\F$ with $K:= \supp (\tilde{\rho} - \rho) \subset \F$ compact such that 
	\[0 < \tilde{\rho}(K) = \rho(K) < \infty \:. \]
	Then ${(\tilde{\rho} - \rho)(\F) = 0}$, i.e.\ \eqref{totvol} is satisfied. 
	Since the Lagrangian is supposed to be of compact range, there exists a compact set $K' \subset \F$ such that $\L(x,y) = 0$ for all~$x \in K$ and~$y \in \F \setminus K'$. 

	%Before going on, we require some additional properties: Let $K, K' \subset \F$ compact.
	By regularity of $\rho$ and~$\tilde{\rho}$, for arbitrary $\tilde{\varepsilon} > 0$ there exist $U \supset K$
	and~$U' \supset K'$ open with~$U' \supset U$ such that
	\begin{align*}
	\rho(U \setminus K) < \tilde{\varepsilon} \:, \qquad \rho(U' \setminus K') < \tilde{\varepsilon} \:, \\
	\tilde{\rho}(U \setminus K) < \tilde{\varepsilon} \:, \qquad \tilde{\rho}(U' \setminus K') < \tilde{\varepsilon} \:.
	\end{align*}
	In particular, $\rho|_{U'}$ is a non-negative finite Borel measure on the topological space $U'$. 
	Since $\F$ is metrizable, it is completely regular. As a consequence, the class $\Gamma_{\rho|_{U'}}$ of all Borel sets $E \subset U'$ with boundaries of $\rho$-measure 
	zero contains a base (consisting of open sets) of the topology of~$U'$ (see~\cite[Proposition 8.2.8]{bogachev2}). Since $K$ and $K'$ are compact, they can be covered by finitely many relatively compact, open sets $V_1, \ldots, V_N \subset U$ and $V_1', \ldots, V_{N'}' \subset U'$ in $\Gamma_{\rho|_{U'}}$ (cf.~\cite[Section 1.1]{engelking}). 
	By construction, the closure of the sets~$V := \bigcup_{i=1}^N V_i \subset U$ and $V' := \bigcup_{i=1}^{N'} V_i' \subset U'$, denoted by $\overline{V}$ and $\overline{V'}$, respectively, is compact. 
	By choosing $V \subset U$ suitably, we can arrange that $\L(x,y) = 0$ for all $x \in V$ and $y \notin V'$. Moreover, considering the restriction $\rho^{(n)}|_V$, for any $f \in C_c(V)$ 
	we obtain   
	\begin{align*}
	\lim_{n \to \infty} \int_V f \: d\rho^{(n)}|_V = \lim_{n \to \infty} \int_{\F} f \: d\rho^{(n)} \stackrel{\eqref{(vague convergence)}}{=} \int_{\F} f \: d\rho = \int_V f \: d\rho|_V \:,
	\end{align*}
	and similarly for $V'$. This proves vague convergence
	\begin{align*}
	\rho^{(n)}|_V \stackrel{v}{\to} \rho|_V \qquad \text{and} \qquad \rho^{(n)}|_{V'} \stackrel{v}{\to} \rho|_{V'} \:. 
	\end{align*}
	Since $\Gamma_{\rho|_{U'}}$ is a subalgebra in $\B(\F)$ (see \cite[Proposition 8.2.8]{bogachev2}), the sets~$V$ and $V'$ are also contained in $\Gamma_{\rho|_{U'}}$, implying that
	\[\rho(\partial V) = 0 = \rho(\partial V') \:. \]
	Making use of vague convergence \eqref{(vague convergence)} and applying \cite[Theorem 30.2]{bauer}, we obtain 
	\begin{align*}
	\rho(V) = \rho(\overline{V}) \ge \limsup_{n \to \infty} \rho^{(n)}(\overline{V}) \ge \limsup_{n \to \infty} \rho^{(n)}(V) \ge \liminf_{n \to \infty} \rho^{(n)}(V) \ge \rho(V) \:, 
	\end{align*}
	proving that 
	\begin{align}\label{(rhoV)}
	\rho(V) = \lim_{n \to \infty} \rho^{(n)}(V)
	\end{align} 
	(and similarly for $V'$). We point out that, for each $n \in \N$, the measure $\rho^{(n)}|_V/\rho^{(n)}(V)$ is normalized in the sense of \S \ref{S Existence compact}. Moreover, for any $f \in C_c(V)$ we are given
	\begin{align*}
	\lim_{n \to \infty} \int_V f \: d\rho^{(n)}|_V/\rho^{(n)}(V) = \int_V f \: d\rho|_V/\rho(V) \:. 
	\end{align*}
	That is, the sequence of normalized measures $(\rho^{(n)}|_V/\rho^{(n)}(V))_{n \in \N}$ converges vaguely to the normalized measure $\rho|_V/\rho(V)$. From \cite[Corollary 30.9]{bauer} we conclude that the sequence $(\rho^{(n)}|_V/\rho^{(n)}(V))_{n \in \N}$ converges \emph{weakly} to the normalized measure $\rho|_V/\rho(V)$. Since $V$ is separable (see e.g.~\cite[Corollary 3.5]{aliprantis}), we may apply \cite[Theorem 2.8]{billingsley} to obtain \emph{weak} convergence of the corresponding product measures, 
	\begin{align}\label{(weak convergence product)}
	\begin{split}
	\big(\rho^{(n)}|_V/\rho^{(n)}(V)\big) \otimes \big(\rho^{(n)}|_V/\rho^{(n)}(V)\big) &\rightharpoonup \big(\rho|_V/\rho(V)\big) \otimes \big(\rho|_V/\rho(V)\big) \:, \\
	\big(\rho^{(n)}|_V/\rho^{(n)}(V)\big) \otimes \big(\rho^{(n)}|_{V'}/\rho^{(n)}(V')\big) &\rightharpoonup \big(\rho|_V/\rho(V)\big) \otimes \big(\rho|_{V'}/\rho(V')\big) \:.
	\end{split}
	\end{align}
	
	We now proceed as follows. First of all, in accordance with \eqref{integrals} we have
	\begin{align*}
	\Sact(\tilde{\rho}) - \Sact(\rho) &= \int_{\F} d(\tilde{\rho} - \rho)(x) \int_{\F} d\rho(y) \:\L(x,y) + \int_{\F} d\rho(x) \int_{\F} d(\tilde{\rho} - \rho)(y) \:\L(x,y) \\
	&\qquad \quad + \int_{\F} d(\tilde{\rho}- \rho)(x) \int_{\F} d(\tilde{\rho}- \rho)(y) \:\L(x,y) \:.
	\end{align*}
	Making use of the symmetry of the Lagrangian and applying Fubini's theorem, we can write this expression as
	\begin{align*}
	\Sact(\tilde{\rho}) - \Sact(\rho) &= 2\int_{\F} d(\tilde{\rho} - \rho)(x) \int_{\F} d\rho(y) \:\L(x,y) \\
	&\qquad \quad + \int_{\F} d(\tilde{\rho}- \rho)(x) \int_{\F} d(\tilde{\rho}- \rho)(y) \:\L(x,y) \:,
	\end{align*}
	and the fact that $\L$ is of compact range yields
	\begin{align*}
	\Sact(\tilde{\rho}) - \Sact(\rho) &= 2\int_K d(\tilde{\rho}- \rho)(x) \int_{K'} d\rho(y) \:\L(x,y) \\
	&\qquad \quad + \int_K d(\tilde{\rho} - \rho)(x) \int_K d(\tilde{\rho}- \rho)(y) \:\L(x,y) \:.
	\end{align*}
	Making use of the fact that 
	$\rho(V \setminus K) < \tilde{\varepsilon}$, $\tilde{\rho}(V \setminus K) < \tilde{\varepsilon}$ and $\rho(V' \setminus K') < \tilde{\varepsilon}$, $\tilde{\rho}(V' \setminus K') < \tilde{\varepsilon}$, we can arrange that, up to an arbitrarily small error term $\varepsilon > 0$, 
	\begin{align*}
	\Sact(\tilde{\rho}) - \Sact(\rho) &\ge 2\int_V d\tilde{\rho}(x) \int_{V'} d\rho(y) \:\L(x,y) - 2\int_V d\rho(x) \int_{V'} d\rho(y) \:\L(x,y) \\
	&\qquad \quad + \int_V d\tilde{\rho}(x) \int_V d\tilde{\rho}(y) \:\L(x,y) - \int_V d\tilde{\rho}(x) \int_V d\rho(y) \:\L(x,y) \\
	&\qquad \quad - \int_V d\rho(x) \int_V d\tilde{\rho}(y) \:\L(x,y) + \int_V d\rho(x) \int_V d\rho(y) \:\L(x,y) - \varepsilon \:. 
	\end{align*} 
	Since $\L \in C_b(\F \times \F)$, 
	by applying weak convergence \eqref{(weak convergence product)} we obtain
	\begin{align*}
	&\Sact(\tilde{\rho}) - \Sact(\rho) \\
	&\qquad \ge \lim_{n \to \infty} \left[2\left(\int_V d\tilde{\rho}(x) \int_{V'} d\rho^{(n)}|_{V'}(y) - \int_V d\rho^{(n)}|_{V}(x) \int_{V'} d\rho^{(n)}|_{V'}(y)\right) \:\L(x,y) \right. \\
	&\qquad \qquad + \int_V d\tilde{\rho}(x) \int_V d\tilde{\rho}(y) \:\L(x,y) - \int_V d\tilde{\rho}(x) \int_V d\rho^{(n)}|_{V}(y) \:\L(x,y) \\
	&\qquad \qquad \left. - \int_V d\rho^{(n)}|_{V}(x) \int_V d\tilde{\rho}(y) \:\L(x,y) + \int_V d\rho^{(n)}|_{V}(x) \int_V d\rho^{(n)}|_{V}(y) \:\L(x,y) \right] - \varepsilon \:,
	\end{align*} 
	or equivalently,
	\begin{align*}
	\Sact(\tilde{\rho}) - \Sact(\rho) &\ge \lim_{n \to \infty} \left[ 2\int_V d\big(\tilde{\rho}- \rho^{(n)}\big)(x) \int_{V'} d\rho^{(n)}(y) \:\L(x,y) \right. \\
	&\qquad \quad \left. + \int_V d\big(\tilde{\rho} - \rho^{(n)}\big)(x) \int_V d\big(\tilde{\rho}- \rho^{(n)}\big)(y) \:\L(x,y) \right] - \varepsilon \:.
	\end{align*}
	
	Next,
	for any~$n \in \N$ we introduce the measures $\tilde{\rho}_n : \B(\F) \to [0, + \infty]$ by 
	\begin{align*}
	\tilde{\rho}_n := \left\{ \begin{array}{cl}
	c_n \:\tilde{\rho} \qquad &\text{on $V$} \\ [0.2em]
	\rho^{(n)} \qquad &\text{on $\F \setminus V$} 
	\end{array} \right. \qquad \text{with} \qquad c_n := \frac{\rho^{(n)}(V)}{\tilde{\rho}(V)} \quad \text{for all $n \in \N$} \:. 
	\end{align*} 
	Considering the compact exhaustion $(K_n)_{n \in\N}$, we thus have $\rho^{(n)}(K_n) = \tilde{\rho}^n(K_n)$ for every $n \in \N$. Furthermore,
	\begin{align}\label{(cn to one')}
	\lim_{n \to \infty} c_n = \lim_{n \to \infty} \frac{\rho^{(n)}(V)}{\tilde{\rho}(V)} \stackrel{\eqref{(rhoV)}}{=} \frac{\rho(V)}{\tilde{\rho}(V)} = \frac{\rho(V \setminus K) + \rho(K)}{\tilde{\rho}(V \setminus K) + \tilde{\rho}(K)} = 1
	\end{align}
	according to $\rho|_{\F \setminus K} = \tilde{\rho}|_{\F \setminus K}$ %(see \S \ref{seccvp}) 
	and $\tilde{\rho}(K) = \rho(K)$. 
	In view of \eqref{(cn to one')} we may write 
	\begin{align*}
	\Sact(\tilde{\rho}) - \Sact(\rho) &\ge \lim_{n \to \infty} \left[ 2\int_V d\big(c_n \: \tilde{\rho}- \rho^{(n)}\big)(x) \int_{V'} d\rho^{(n)}(y) \:\L(x,y) \right. \\
	&\qquad \quad \left. + \int_V d\big(c_n \: \tilde{\rho} - \rho^{(n)}\big)(x) \int_V d\big(c_n \: \tilde{\rho}- \rho^{(n)}\big)(y) \:\L(x,y) \right] - \varepsilon \:.
	\end{align*}
	Since $\tilde{\rho}_n$ and $\rho^{(n)}$ coincide on $K_n \setminus V$ for all sufficiently large~$n \in \N$, and $\L(x,y) = 0$ for all $x \in V$ and~$y \notin V'$, the difference $\Sact(\tilde{\rho}) - \Sact(\rho)$ can finally be written as
	\begin{align*}
	\Sact(\tilde{\rho}) - \Sact(\rho) &= \lim_{n \to \infty} \left[2 \int_{K_n} d\big(\tilde{\rho}_n- \rho^{(n)}\big)(x) \int_{K_n} d\rho^{(n)}(y) \:\L(x,y) \right. \\
	&\qquad \qquad \left. + \int_{K_n} d\big(\tilde{\rho}_n - \rho^{(n)}\big)(x) \int_{K_n} d\big(\tilde{\rho}_n - \rho^{(n)}\big)(y) \:\L(x,y) \right] - \varepsilon \:. 
	\end{align*}
	Since $\rho^{(n)}$ is a minimizer on $K_n$ for every $n \in \N$, we are given
	\begin{align}\label{(ge')}
	\big(\Sact_{K_n}(\tilde{\rho}_{n}) - \Sact_{K_n}(\rho^{(n)})\big) \ge 0 \qquad \text{for all $n \in \N$} \:.
	\end{align}
	Taking the limit $n \to \infty$ on the left hand side of \eqref{(ge')}, one obtains exactly the above expression in square brackets, which therefore is bigger than or equal to zero. Since~$\varepsilon > 0$ is arbitrary, we thus arrive at 
	\begin{align*}
	\big(\Sact(\tilde{\rho}) - \Sact(\rho)\big) \ge 0 \:.
	\end{align*}
	Hence $\rho$ is a minimizer under variations of compact support.  
\end{proof}

\subsection{Existence of Minimizers under Variations of Finite Volume} \label{secvaryfinite}
In order to prove the existence of minimizers in the sense of Definition \ref{Definition minimizer}, we additionally assume that property (iv) in~\S\ref{seccvp} is satisfied, i.e. 
\begin{align*}
	\sup_{x \in \F} \int_{\F} \L(x,y) \:d\rho(y) < \infty \:.
\end{align*}
Under this additional assumption, the difference \eqref{integrals} is well-defined. Moreover, we obtain the following existence result.

\begin{Thm}\label{Theorem global}
	Let $\L : \F \times \F \to \R_0^+$ be continuous, bounded, and of compact range, and assume that condition~{\rm{(iv)}} in~\S\ref{seccvp} is satisfied. Then $\rho$ is a minimizer under variations of finite volume (see Definition~\ref{Definition minimizer}).
\end{Thm}
\begin{proof}
Let $\tilde{\rho} \in \mathfrak{B}_{\F}$ be a positive Borel measure on $\F$ satisfying \eqref{totvol}, i.e.
\begin{align*}
|\tilde{\rho} - \rho|(\F) < \infty \qquad \text{and} \qquad (\tilde{\rho} - \rho) (\F) = 0 \:.
\end{align*}
Introducing $B := \supp (\tilde{\rho} - \rho)$, we thus are given $0 < \rho(B) = \tilde{\rho}(B) < \infty$. 
Without loss of generality, we may assume that $\rho(B) = \tilde{\rho}(B) > 1$ (otherwise we multiply the measures $\rho$ and $\tilde{\rho}$ by a suitable constant). 
By assuming that condition~{\rm{(iv)}} in~\S\ref{seccvp} holds we know that the difference~\eqref{integrals} is well-de\-fined, thus giving rise to 
\begin{align*}
\Sact(\tilde{\rho}) - \Sact(\rho) 
&= 2\int_{B} d(\tilde{\rho} - \rho)(x) \int_{\F} d\rho(y) \: \L(x,y) \\  
&\qquad \qquad 
+ \int_{B} d(\tilde{\rho}- \rho)(x) \int_{B} d(\tilde{\rho}- \rho)(y) \: \L(x,y) \:.
\end{align*}
Let $\tilde{\varepsilon} > 0$ be arbitrary. 
In analogy to the proof of Theorem \ref{Theorem global}, by regularity of $\rho$ and~$\tilde{\rho}$ we approximate~$B$ by open sets $U \supset B$ from outside such that 
\begin{align*}
\rho(U \setminus B) < \tilde{\varepsilon}/4 \:, \qquad \tilde{\rho}(U \setminus B) < \tilde{\varepsilon}/4 \:.
\end{align*}
By adding and subtracting the terms
\begin{align*}
2 \int_{U \setminus B} d(\tilde{\rho} - \rho)(x) \int_{\F} d\rho(y) \: \L(x,y) + \int_{U \setminus B} d(\tilde{\rho} - \rho)(x) \int_B d(\tilde{\rho} - \rho)(y) \: \L(x,y)
\end{align*}
as well as
\[ \int_{U} d(\tilde{\rho} - \rho)(x) \int_{U \setminus B} d(\tilde{\rho} - \rho)(y) \: \L(x,y) \:, \]
one can show that
\begin{align*}
&\big(\Sact(\tilde{\rho}) - \Sact(\rho) \big) = 2 \int_{U} d(\tilde{\rho} - \rho)(x) \int_{\F} d\rho(y) \: \L(x,y) \\
&\quad + \int_{U} d(\tilde{\rho} - \rho)(x) \int_U d(\tilde{\rho} - \rho)(y) \: \L(x,y) - \left\{ \int_{U} d(\tilde{\rho} - \rho)(x) \int_{U \setminus B} d(\tilde{\rho} - \rho)(y) \: \L(x,y) \right. \\
&\quad + \left. 2 \int_{U \setminus B} d(\tilde{\rho} - \rho)(x) \int_{\F} d\rho(y) \: \L(x,y) + \int_{U \setminus B} d(\tilde{\rho} - \rho)(x) \int_B d(\tilde{\rho} - \rho)(y) \: \L(x,y)\right\} \:. 
\end{align*}
Choosing~$U \supset B$ suitably, property (iv) implies (along with \eqref{(Lagrange functional)}) that the expression
\begin{align*}
\left|\int_{U \setminus B} d(\tilde{\rho} - \rho)(x) \int_{\F} d\rho(y) \: \L(x,y) \right| \le \underbrace{\left(\sup_{x \in \F} \ell(x) + 1\right)}_{\text{$< \infty$}} \underbrace{\big(\left|\tilde{\rho}(U \setminus B)\right| + \left|\rho(U \setminus B) \right|\big)}_{\text{$< \tilde{\varepsilon}/2$}}
\end{align*}
can be arranged to be arbitrarily small. Similarly, one can show that the other summands in the expression in curly brackets are arbitrarily small 
for a suitable choice of~$U \supset B$. We thus can arrange that
\begin{align*}
\Sact(\tilde{\rho}) - \Sact(\rho) &\ge  2 \int_{U} d(\tilde{\rho} - \rho)(x) \int_{\F} d\rho(y) \:\L(x,y) \\
& \qquad \qquad 
+ \int_{U} d(\tilde{\rho}- \rho)(x) \int_{U} d(\tilde{\rho}- \rho)(y) \:\L(x,y) - \varepsilon
\end{align*} 
for any given $\varepsilon > 0$. Next, by regularity of $\rho$ and $\tilde{\rho}$ we may approximate~$U$ by compact sets $V \subset U$ from inside 
such that
\begin{align*}
\rho(U \setminus V) < \tilde{\varepsilon}/4 \:, \qquad \tilde{\rho}(U \setminus V) < \tilde{\varepsilon}/4 \:.
\end{align*}
Proceeding in analogy to the proof of Theorem~\ref{Lemma compact minimizer}, we may assume that $V \subset U$ is an open, relatively compact set such that $\rho^{(n)}(V) \to \rho(V)$. 
Moreover, since $\L$ is supposed to be of compact range, there is some relatively compact, open subset~$V' \subset \F$ such that~$\L(x,y) = 0$ for all $x \in V$ and~$y \notin V'$ (and vice versa, for details see the proof of Theorem~\ref{Lemma compact minimizer}). Applying similar arguments as before to $V \subset U$, we arrive at 
\begin{align*}
\Sact(\tilde{\rho}) - \Sact(\rho) &\ge \left[ 2 \int_{V} d(\tilde{\rho} - \rho)(x) \int_{V'} d\rho(y) \:\L(x,y) \right. \\
& \qquad \qquad 
\left. + \int_{V} d(\tilde{\rho}- \rho)(x) \int_{V} d(\tilde{\rho}- \rho)(y) \:\L(x,y) \right] - 2 \varepsilon \:. 
\end{align*} 
Proceeding in analogy to the proof of Theorem~\ref{Lemma compact minimizer}, one can show that the term in square brackets is greater than or equal to zero, up to an arbitrarily small error term. Indeed, introducing the measures 
\begin{align*}
\tilde{\rho}_n := \left\{ \begin{array}{cl}
c_n \:\tilde{\rho} \qquad &\text{on $V$} \\ [0.2em]
\rho^{(n)} \qquad &\text{on $\F \setminus V$} 
\end{array} \right. \qquad \text{with} \qquad c_n := \frac{\rho^{(n)}(V)}{\tilde{\rho}(V)} \quad \text{for all $n \in \N$} \:,
\end{align*} 
%it follows that~$\supp \big(\tilde{\rho}^n - \rho^{(n)} \big) \subset K$ is compact. Moreover,
we are given
\begin{align*}
\tilde{\rho}_n(V) = c_n \: \tilde{\rho}(V) = \frac{\rho^{(n)}(V)}{\tilde{\rho}(V)} \: \tilde{\rho}(V) = \rho^{(n)}(V) \:.
\end{align*}
Note that $\rho(V), \tilde{\rho}(V) \in (\rho(B)- \tilde{\varepsilon}/2, \rho(B)+ \tilde{\varepsilon}/2) \subset (1, \infty)$ by choosing $\tilde{\varepsilon} > 0$ suitably, implying that $1 < \rho(V) < \tilde{\rho}(V) + \tilde{\varepsilon}$. 
%Without loss of generality, we may assume that $\tilde{\rho}(V) \ge 1$. 
In view of~$\rho^{(n)}(V) \to \rho(V)$, for sufficiently large~$n \in \N$ we thus obtain 
\begin{align*}
c_n = \frac{\rho^{(n)}(V)}{\tilde{\rho}(V)} < \frac{\tilde{\rho}(V) + \tilde{\varepsilon}}{\tilde{\rho}(V)} = 1 + \frac{\tilde{\varepsilon}}{\tilde{\rho}(V)} \le 1 + \tilde{\varepsilon} 
\end{align*}
(and similarly, $c_n > 1- \tilde{\varepsilon}$ for sufficiently large~$n \in \N$). 
As a consequence, $$c_n^2 < (1 + \tilde{\varepsilon})^2 < 1 + 3 \tilde{\varepsilon} \:, $$ and henceforth the term in square brackets can be estimated by 
\begin{align*}
&\left(2 \int_{V} d\big(\tilde{\rho} - \rho^{(n)}\big)(x) \int_{V'} d\rho^{(n)}(y) + \int_{V} d\big(\tilde{\rho}- \rho^{(n)}\big)(x) \int_{V'} d\big(\tilde{\rho}- \rho^{(n)}\big)(y)\right) \:\L(x,y) \\
&\qquad \ge \left\{ 2 \int_{V} d\big(c_n \: \tilde{\rho} - \rho^{(n)}\big)(x) \int_{V'} d\rho^{(n)}(y) \:\L(x,y) \right. \\
&\qquad \qquad \left. + \int_{V} d\big(c_n \: \tilde{\rho} - \rho^{(n)}\big)(x) \int_{V'} d\big(c_n \: \tilde{\rho} - \rho^{(n)}\big)(y) \:\L(x,y) \right\} \\
&\qquad \qquad - 4 \tilde{\varepsilon} \: \int_{V} d\tilde{\rho}(x) \int_{V'} d\rho^{(n)}(y) \: \L(x,y) - 3 \tilde{\varepsilon} \: \int_V d\tilde{\rho} (x) \int_{V'} d\tilde{\rho} (y) \: \L(x,y) \:.
\end{align*}
Choosing $\tilde{\varepsilon} > 0$ suitably, one can arrange that 
\begin{align*}
\Sact(\tilde{\rho}) - \Sact(\rho) &\ge \lim_{n \to \infty} \left\{ 2 \int_{V} d\big(c_n \: \tilde{\rho} - \rho^{(n)}\big)(x) \int_{V'} d\rho^{(n)}(y) \:\L(x,y) \right. \\
&\qquad  \qquad+ \left. \int_{V} d\big(c_n \: \tilde{\rho}- \rho^{(n)}\big)(x) \int_{V'} d\big(c_n \: \tilde{\rho}- \rho^{(n)}\big)(y) \:\L(x,y) \right\} - 9 \varepsilon \:.
\end{align*}
Note that $V' \subset K_n$ for sufficiently large~$n \in \N$. Making use of the fact that~$\rho^{(n)}$ is a minimizer on $K_n$ and arguing as in the proof of Theorem~\ref{Lemma compact minimizer}, we conclude that the term in
curly brackets in the last inequality is greater than or equal to zero.

Finally, since $\varepsilon > 0$ was chosen arbitrarily, we arrive at
\begin{align*}
\Sact(\tilde{\rho}) - \Sact(\rho) \ge 0 \:,
\end{align*}
which proves the claim. 
\end{proof}

\section{Minimizers for Lagrangians Decaying in Entropy} \label{Chapter 5}
\subsection{Preliminaries}
The goal of this section is to deal with the question if it is possible to weaken the assumption that the Lagrangian~$\L$ is of compact range. To this end, we specialize the above setting as follows. As before, we let $\F$ be a second-countable, locally compact Hausdorff space. Then $\F$ is completely metrizable, and hence can be endowed with a Heine-Borel metric as mentioned in \S\ref{S Basic Definitions} such that 
$\F$ is proper, i.e.\ closed, bounded subsets in $\F$ are compact. As every relatively compact set is precompact, any bounded subset of $\F$ can be covered by a finite number of sets of diameter less than~$\delta > 0$ (cf.~
\cite[\S 3.16, \S 3.17]{dieudonne1}). Thus for any $r > 0$ and $x \in \F$,
the closed ball $\overline{B_r(x)}$ is compact, and hence can be covered by finitely many balls of radius~$\delta > 0$. We denote the smallest such number by~$E_x(r, \delta)$.\footnote{In coarse geometry, this number is called \emph{entropy} of a set (cf.~\cite[Definition 3.1]{roecoarse}). In the literature, however, also the logarithm of this number is referred to as \emph{$\delta$-entropy} (see~\cite[\S 3.16, Problem~4]{dieudonne1}).} 
In particular, for all~$r' < r$ the annuli~$\overline{B_r(x) \setminus B_{r'}(x)}$ can be covered by at most~$E_x(r, \delta)$ balls of radius~$\delta$. If $\rho$ is a uniform measure on $\F$, the number $E_x(r, \delta)$ can be determined more specifically (see~\cite[Example 3.13]{roecoarse}). 

In the following, we additionally assume that the Lagrangian decays in entropy, which is defined as follows. 

\begin{Def}\label{Definition vanishing}
Assume that $\F$ is endowed with an unbounded Heine-Borel metric~$d$.
The Lagrangian $\L \colon \F \times \F \to \R_0^+$ is said to {\bf{decay in entropy}}
if the following conditions are satisfied:
	\begin{enumerate}[leftmargin=2em]
		\item[\rm{(a)}] $c:=\inf_{{x} \in \F} \L({x},{x}) > 0$.
		\item[\rm{(b)}] There is a compact set $K \subset \F$ such that 
		\[ \delta := \inf_{{x} \in \F \setminus K} \sup \left\{ s \in \R \::\: \L({x},y) \ge \frac{c}{2} \quad \text{for all~$y \in B_s(x)$} \right\} > 0 \:. \]
		\item[\rm{(c)}] The Lagrangian has the following decay property:
		There is a monotonically decreasing, integrable function~$f \in L^1(\R^+, \R^+_0)$ such that
		\begin{align*}
		\L(x,y) \leq \frac{f \big(d(x,y) \big)}{C_x\big(d(x,y), \delta \big)} \qquad \text{for all~$x,y \in \F$ with~$x \neq y$} \:,
		\end{align*}
where
		\[C_x(r, \delta) := C \: E_x(r+2, \delta) \qquad \text{for all $r > 0$}\:, \]
		and the constant $C$ is given by
		\[ C := 1+\frac{2}{c} < \infty \:. \]
	\end{enumerate}
\end{Def}\noindent
In Definition~\ref{Definition vanishing}~(b) we may assume without loss of generality that $\delta = 1$ (otherwise we
rescale the metric suitably). Then  
\[\L(x,y) \ge \frac{\L(x,x)}{2} \qquad \text{for all $y \in B_1(x)$} \:. \]
Now let $(\rho^{(n)})_{n \in \N}$ be the sequence of measures given by \eqref{rhorund}, and let~$\rho$ be its vague limit constructed in~\eqref{rhoglobal}. Then by \eqref{(lemmaupper)}, for every $x \in \F$ we have $\rho(B_1(x)) \le C$ as well as~$\rho^{(n)} (B_1(x)) \le C$ for sufficiently large $n \in \N$.

Condition~(b) determines the behavior of the Lagrangian locally (more precisely, it gives a uniform bound
for the size of balls in which the Lagrangian is bounded from below).
In condition~(c), on the other hand, the function~$f$ characterizes the decay properties of the Lagrangian
at infinity. In particular, condition~(c) implies that for any~$\varepsilon > 0$ there is an integer~$N_0=N_0(\varepsilon) > 1$ such that
\beq \label{(less e)}
\sum_{k= n}^{\infty} f(k) \leq \int_{n-1}^\infty f(x)\: dx < \varepsilon/3 \qquad \text{for all $n \ge N_0$}\:. 
\eeq

Considering arbitrary $\varepsilon > 0$ and $x\in \F$, the Heine-Borel property of $\F$ ensures that the closed ball
\beq\label{Kxe}
K_{x, \varepsilon} := \overline{B_{N_0}(x)} 
\eeq
is compact. Since $\L$ decays in entropy, we thus obtain
\begin{align*}
&\int_{\F \setminus K_{x, \varepsilon}} \L(x,y) \:d\rho(y) = \sum_{k=N_0}^{\infty} \int_{{B_{k+1}(x)} \setminus \overline{B_{k}(x)}} \L(x,y) \:d\rho(y) \\
&\;\;\le \sum_{k=N_0}^{\infty} \sup_{y \in {B_{k+1}(x)} \setminus \overline{B_{k}(x)}} \!\!\! \L(x,y) \:\underbrace{\rho\left({B_{k+1}(x)} \setminus \overline{B_{k}(x)}\right)}_{\text{$\le C_x(k,1)$}} \le \sum_{k=N_0}^{\infty} \frac{f(k)}{C_x(k, 1)} \:C_x(k,1) < \varepsilon/3 \:,
\end{align*}
where in the last step we made use of \eqref{(less e)}. Analogously, 
\beq
\int_{\F \setminus K_{x, \varepsilon}} \L(\tilde{x}, y) \:d\rho(y) < \varepsilon/3 \label{eps3}
\eeq
for all $\tilde{x} \in B_{\delta}(x)$ and sufficiently small~$\delta > 0$. 
Under the assumption that the Lagrangian is continuous, the same is true for the measures~$\rho^{(n)}$:
Given~$x \in \F$, we introduce the compact sets~$A_k(x) \subset \F$ by   
\[A_k(x) := \overline{{B_{k+1}(x)} \setminus {B_{k}(x)}} \qquad \text{for all $k \ge N_0$} \]
%the closure $\overline{A_k}$ is compact for all $k \in \N$. C
and choose open sets $U_k(x) \supset {A_k}(x)$ with $U_k(x) \subset B_{k+2}(x) \setminus B_{k-1}(x)$ such that
\begin{align*}
	\rho \left(U_k \setminus \big({B_{k+1}(x)} \setminus \overline{B_{k}(x)}\big)\right) < 2^{-k-1} \: \varepsilon/3 \qquad \text{for all $k \ge N_0$} \:.
\end{align*}
Then in view of~\cite[Lemma 2.92]{aliprantis}, for every~$k \ge N_0$ there exists $\eta_k \in C_c(U_k(x); [0,1])$ such that~$\eta_k|_{{A_k(x)}} \equiv 1$, 
%Given~$x \in \F$, 
%and $\delta > 0$ sufficiently small, for all~$\tilde{x} \in B_{\delta}(x)$ 
%we thus obtain
implying that~$\L(x, .) \: \eta_k \in C_c(U_k(x))$ for all~$k \ge N_0$. 
%For every~$\tilde{x} \in B_{\delta}(x)$ and 
Hence for arbitrary~$m > 0$, by vague convergence~\eqref{(vague convergence)} there exists some integer~$N_k = N_k(m)$ such that
\begin{align}\label{(decay)}
	\left|\int_{U_k(x)} \L(x,y) \: \eta_k(y) \: d\rho^{(n)}(y) - \int_{U_k(x)} \L(x,y) \: \eta_k(y) \: d\rho(y)\right| \le 2^{-k}/m
\end{align}
for all $n \ge N_k$. For $k = N_0$, we choose a subsequence of $(\rho^{(n)})_{n \in \N}$ which we denote by~$(\rho^{(1,n_{\ell})})_{\ell \in \N}$ such that \eqref{(decay)} holds for every $m \in \N$ and all $\ell \ge N(m) := N_1(m)$. Next, for $k= N_0+1$ we choose a subsequence of $(\rho^{(1,n_{\ell})})_{\ell \in \N}$, denoted by $(\rho^{(2,n_{\ell})})_{\ell \in \N}$, such that \eqref{(decay)} is satisfied for every $m \in \N$ and all $\ell \ge N(m)$. Proceeding iteratively, we obtain a corresponding diagonal sequence $(\rho^{(\ell, n_{\ell})})_{{\ell} \in \N}$, which for convenience we shall again denote by $(\rho^{(n)})_{n \in \N}$. This sequence has the property that, for each $k \ge N_0$ and for every~$m \in \N$, there exists $N = N(m)$ such that \eqref{(decay)} is satisfied for all $n \ge N$. As a consequence, for arbitrary $\varepsilon > 0$ there exists $m \in \N$ such that $1/m < \varepsilon/6$, 
giving rise to
%and for all $\tilde{x}$ in a small neighborhood of $x$ and arbitrary $m > 0$, we obtain
\begin{align*}
&\int_{\F \setminus K_{x, \varepsilon}} \L({x},y) \:d\rho^{(n)}(y) = \sum_{k=N_0}^{\infty} \int_{{B_{k+1}(x)} \setminus \overline{B_{k}(x)}} \L({x},y) \:d\rho^{(n)}(y) \\
&\;\le \sum_{k=N_0}^{\infty} \int_{U_k(x)} \L({x},y) \: \eta_k(y) \:d\rho^{(n)}(y) %= \sum_{k=N_0}^{\infty} \lim_{n \to \infty} \int_{U_k(x)} \L({x},y) \: \eta_k(y) \:d\rho^{(n)}(y) \\
% &\;\;
\stackrel{\eqref{(decay)}}{\le} \sum_{k=N_0}^{\infty} \left(\int_{U_k} \L({x},y) \: \eta_k(y) \:d\rho(y) + 2^{-k}/m \right) 
%\le \sum_{k=N_0}^{\infty} \int_{A_{k+1}(x)} \L({x},y) \: \eta_k(y) \:d\rho(y) + \frac{1}{m}
% \\ &\;\;\le \sum_{k=N_0}^{\infty} \sup_{y \in {B_{k+1}(x)} \setminus \overline{B_{k}(x)}} \!\!\! \L(x,y) \:\underbrace{\rho\left({B_{k+1}(x)} \setminus \overline{B_{k}(x)}\right)}_{\text{$\le C_x(k,1)$}} \le \sum_{k=N_0}^{\infty} \frac{f(k)}{C_x(k, 1)} \:C_x(k,1) 
< \varepsilon
\end{align*}
for all $n \ge N(m)$, 
where we applied the geometric series. 

Next, we employ the fact that $\F$ is separable. Consequently, there exists a countable dense subset $\{x_i : i \in \N \}$. Then the above arguments show that 
\begin{align}\label{(neighborhood)}
	\int_{\F \setminus K_{x, \varepsilon}} \L({x},y) \:d\rho^{(n)}(y) < \varepsilon \qquad \text{for all $n \ge N(m)$} 
\end{align}
holds for $x = x_1$. By choosing a suitable subsequence $(\rho^{(2, n_{\ell})})_{\ell \in \N}$, we can arrange that~\eqref{(neighborhood)} holds for $x = x_{2}$. Proceeding iteratively, we end up with a corresponding diagonal sequence $(\rho^{(\ell, n_{\ell})})_{\ell \in \N}$, which for simplicity we again denote by $(\rho^{(n)})_{n \in \N}$, such that~\eqref{(neighborhood)} holds for $x = x_i$ for every $i \in \N$.
%We now proceed iteratively as follows. Applying the same arguments as before to $x_{i+1}$, we obtain a 
%Given $x, y \in \F$ and $\delta > 0$ sufficiently small, continuity of the Lagrangian implies that the restriction
%$$\L|_{\overline{B_{\delta}(x)} \times \overline{B_{\delta}(y)}} : \overline{B_{\delta}(x)} \times \overline{B_{\delta}(y)} \to \R_0^+$$ is uniformly continuous. As a consequence, the mapping 
%$$\tilde{\L}_{\delta}(x, .) := \sup_{\tilde{x} \in B_{\delta}(x)} \L(\tilde{x}, .) : \F \to \R_0^+$$ is continuous, and thus $\tilde{\L}_{\delta}(x, .) \: \eta_k \in C_c(U_k(x))$ for all $k \ge N_0$. From this we 
%, and interchanging the limit and the sum is justified in view of 
%Making use of continuity of the Lagrangian, we finally 
Combining the above inequalities, we obtain
\[ \int_{\F \setminus K_{x, \varepsilon}} \L(\tilde{x},y) \:d\rho^{(n)}(y) \leq 
\sum_{k=N_0}^{\infty} \left(\int_{U_k} \L(\tilde{x},y) \: \eta_k(y) \:d\rho(y) + 2^{-k}/m \right) \:. \]
Using~\eqref{eps3}, we conclude that
for all~$\tilde{x}$ in a small neighborhood of~$x$ and for sufficiently large $n \in \N$,
\begin{align}\label{(less epsilon)}
\int_{\F \setminus K_{x, \varepsilon}} \L(\tilde{x},y) \:d\rho^{(n)}(y) < \varepsilon \qquad \text{and} \qquad  \int_{\F \setminus K_{x, \varepsilon}} \L(\tilde{x},y) \:d\rho(y) < \varepsilon \:.
\end{align}
In particular, these inequalities imply that
the measure~$\rho$ constructed in~\eqref{rhoglobal} is non-zero
(for details see Lemma~\ref{Lemma nontrivial decay}).

\subsection{Preparatory Results}
Based on \eqref{(less epsilon)}, the goal of this subsection is to derive 
results similar to Lemma~\ref{Lemma Continuity of ell}, Proposition~\ref{Proposition Pointwise Convergence Locally Compact} and Proposition~\ref{Proposition Equicontinuity Locally Compact}.
We first prove continuity of the function $\ell$ (as defined in \eqref{(Lagrange functional)}). 

\begin{Prp}\label{Proposition continuity'}
	Assume that the Lagrangian $\L : \F \times \F \to \R_0^+$ is continuous and decays in entropy. 
	Then the function~$\ell : \F \to \R$ is continuous.
\end{Prp}
\begin{proof}
	Let $x \in \F$ and $\varepsilon > 0$ be arbitrary, and let $(x_n)_{n \in \N}$ be an arbitrary sequence in~$\F$ converging to $x$. Introducing the associated compact set~$K_{x, \varepsilon}$ (as defined in \eqref{Kxe}), by continuity of $\L$ and $\rho(K_{x, \varepsilon}) < \infty$ we obtain
	\begin{align*}
	&\big|\ell(x) - \ell(x_n) \big| \le \bigg|\int_{K_{x, \varepsilon}} \!\!\! \big(\L(x,y) - \L(x_n, y) \big) \:d\rho(y)\bigg| + \bigg|\int_{\F \setminus K_{x, \varepsilon}} \!\!\! \big(\L(x,y) - \L(x_n, y) \big) \:d\rho(y) \bigg| \\
	&\qquad \le  \big|\L(x,y) - \L(x_n, y) \big| \: \rho(K_{x,\varepsilon}) + \left|\int_{\F \setminus K_{x, \varepsilon}} \!\!\! \big(\L(x,y) - \L(x_n, y) \big) \:d\rho(y)\right| \overset{\eqref{(less epsilon)}}{<} 3 \varepsilon 	
	\end{align*}
	for sufficiently large $n \in \N$. This proves continuity of $\ell$. 
\end{proof}

\begin{Prp}\label{Proposition pointwise'}
	Let $(\ell^{(n)})_{n \in \N}$ and $\ell$ be the functions defined in \eqref{lndef'} and \eqref{(Lagrange functional)}. Then $(\ell^{(n)})_{n \in \N}$ converges pointwise to $\ell$, i.e.
\[ 
	\lim_{n \to \infty} \ell^{(n)}(x) = \ell(x) \qquad \text{for all $x \in \F$} \:. \]
\end{Prp}
\begin{proof}
	Let $x \in \F$ and $\varepsilon > 0$ be arbitrary, and let~$K_{x,\varepsilon} \subset \F$ compact according to~\eqref{Kxe}. 
	Given~$U \supset K_{x, \varepsilon}$ open, there exists $\eta \in C_c(U;[0,1])$ such that $\eta|_{K_{x, \varepsilon}} \equiv 1$ (see e.g.~\cite[Lemma~2.92]{aliprantis}). 
	Thus in view of~\eqref{(less epsilon)}, for sufficiently large $n \in \N$ we obtain
	\begin{align*}
	\left| \int_{U \setminus K_{x,\varepsilon}} \L(x,y) \: \eta(y) \:d\rho^{(n)}(y) \right| \le \left| \int_{\F \setminus K_{x,\varepsilon}} \L(x,y) \: \:d\rho^{(n)}(y) \right| \stackrel{\eqref{(less epsilon)}}{<} \varepsilon \:.
	\end{align*}
	Making use of $\L(x, .) \: \eta \in C_c(\F)$ and vague convergence \eqref{(vague convergence)}, this yields 
	\begin{align*}
	&\left|\ell(x) - \ell^{(n)}(x) \right| = \left|\int_{\F} \L(x,y) \:d\rho(y) - \int_{\F} \L(x,y) \:d\rho^{(n)}(y) \right| \\
	&\qquad \le \left|\int_{K_{x,\varepsilon}} \L(x,y) \: \eta(y) \:d\big(\rho - \rho^{(n)}\big)(y) \right| + \underbrace{\left|\int_{\F \setminus K_{x,\varepsilon}} \L(x,y) \:d\big(\rho - \rho^{(n)}\big)(y)\right|}_{\text{$< 2\varepsilon$}} \\
	&\qquad < \left|\int_{U} \L(x,y) \: \eta(y) \:d\big(\rho - \rho^{(n)}\big)(y) \right| + \underbrace{\left|\int_{U \setminus K_{x,\varepsilon}} \L(x,y) \: \eta(y) \:d\big(\rho - \rho^{(n)}\big)(y) \right|}_{\text{$< 2 \varepsilon$}} + \: {2 \varepsilon} \\
	&\qquad < \varepsilon + 2 \varepsilon + 2 \varepsilon < 5 \varepsilon
	\end{align*}
	for sufficiently large $n \in \N$,
	which gives the claim. 
\end{proof}
%\begin{proof}

\begin{Prp}\label{Proposition equicontinuity'}
	Assume that the Lagrangian $\L : \F \times \F \to \R$ is continuous and decays in entropy, and let~$K \subset \F$ be compact. Then for every $x \in K$ and every sequence~$\big(x^{(n)}\big)_{n \in \N}$ in~$K$ with $x^{(n)} \to x$ we have
\[ 
	\lim_{n \to \infty} \left|\ell^{(n)}|_K \big(x^{(n)}\big) - \ell^{(n)}|_K (x) \right| = 0 \:. \]
\end{Prp}
\begin{proof}
	Let $K \subset \F$ be a compact subset. For any $x \in K$ and $\varepsilon >0$, there is a compact subset~$K_{x, \varepsilon} \subset \F$ (defined by~\eqref{Kxe}) such that \eqref{(less epsilon)} is satisfied. Let~$C(x, \varepsilon) > 0$ be the positive constant according to Lemma~\ref{lemmaupper} such that $\rho^{(n)}(K_{x,\varepsilon}) \le C(x, \varepsilon)$ for all~$n \in \N$. Since~$\L$ is continuous and~$K \times K_{x, \varepsilon}$ is compact, the mapping
	\begin{align*}
	\L|_{K \times K_{x,\varepsilon}} : K \times K_{x, \varepsilon} \to \R
	\end{align*}
	is uniformly continuous. Hence we may choose $\delta > 0$ such that
	\[\left| \L|_{K \times K_{x, \varepsilon}} (x, \cdot) - \L|_{K \times K_{x,\varepsilon}}(z, \cdot) \right| < \frac{\varepsilon}{2 C(x, \varepsilon)} \qquad \text{for all $z \in B_{\delta}(x) \cap K$} \:. \]
	In view of \eqref{(less epsilon)}, for sufficiently large $n \in \N$ we thus obtain
	\begin{align*}
	& \sup_{z \in B_{\delta}(x) \cap K} \left|\ell^{(n)}|_K(x) - \ell^{(n)}|_K (z) \right| =\!\! \sup_{z \in B_{\delta}(x) \cap K} \left|\int_{\F} \big(\L|_{K \times \F}(x,y) - \L|_{K \times \F}(z,y) \big) \:d\rho^{(n)}(y) \right| \\
	&\qquad \le \sup_{z \in B_{\delta}(x) \cap K} \left|\int_{\F \setminus K_{x, \varepsilon}} \big(\L|_{K \times \F}(x,y) - \L|_{K \times \F}(z,y) \big) \:d\rho^{(n)}(y) \right| \\
	&\qquad \qquad + \sup_{z \in B_{\delta}(x) \cap K} \left|\int_{K_{x, \varepsilon}} \big(\L|_{K \times \F}(x,y) - \L|_{K \times \F}(z,y) \big) \:d\rho^{(n)}(y) \right| < \frac{\varepsilon}{2} + \frac{\varepsilon}{2} = \varepsilon \:.
	\end{align*}
	Considering a sequence $\big(x^{(n)}\big)_{n \in \N}$ in $K$ with $x^{(n)} \to x$, we have
	\begin{align*}
	\lim_{n \to \infty} \left|\ell^{(n)}|_K \big(x^{(n)}\big) - \ell^{(n)}|_K (x) \right| = 0 \:,
	\end{align*}
	which completes the proof. 
\end{proof}

\subsection{The Euler-Lagrange Equations}
Now we are able to prove the EL equations in the case that $\L$ decays in entropy
(see Definition~\ref{Definition vanishing}). 

\begin{Thm}\label{Theorem EL vanishing}
	Assume that $\L$ is continuous and decays in entropy. 
	Then the measure~$\rho$ constructed in \eqref{rhoglobal} satisfies the Euler-Lagrange equations
	\begin{align*}
	\ell|_{\supp \rho} \equiv \inf_{x \in \F} \ell(x) = 0 \:,
	\end{align*}
	where $\ell \in C(\F)$ is defined by \eqref{(Lagrange functional)}. 
\end{Thm}
\begin{proof}
	Proceed in analogy to the proof of Theorem \ref{Theorem EL equations locally compact}, and make use of Proposition~\ref{Proposition continuity'}, Proposition \ref{Proposition pointwise'}, and Proposition \ref{Proposition equicontinuity'}. 
\end{proof}

We now generalize Lemma~\ref{Lemma conditions}.
\begin{Lemma}
	Assume that $\L : \F \times \F \to \R$ is continuous and decays in entropy. 
	Under the additional assumptions~{\rm{(a)}--\rm{(c)}} in Lemma \ref{Lemma conditions} (with~$K_x$ replaced by~$K_{x,\varepsilon}$), 
	the measure~$\rho$ constructed in~\eqref{rhoglobal} satisfies property~{\rm{(iv)}} in~\S\ref{seccvp}.	
\end{Lemma}
\begin{proof}
Since continuity of~$\L$ implies that
	\begin{align*}
	\int_{\F} \L(x,y) \: d\rho(y) = \int_{\F \setminus K_{x, \varepsilon}} \L(x,y) \: d\rho(y) + \int_{K_{x, \varepsilon}} \L(x,y) \: d\rho(y) < \infty \:,
	\end{align*}
	the function $\L(x, \cdot) : \F \to \R$ is $\rho$-integrable for every $x \in \F$. It remains to show that
	\begin{align*}
	\sup_{x \in \F} \int_{\F} \L(x,y) \:d\rho(y) < \infty \:.
	\end{align*}  
	For all $x \in M := \supp \rho$ this result follows from Theorem \ref{Theorem EL vanishing}. Whenever $x \in \F \setminus M$, similar as in the proof of Lemma \ref{Lemma conditions} we obtain
	\begin{align*}
	\sup_{x \in \F \setminus M} \int_{\F} \L(x,y) \: d\rho(y) &\le \sup_{x \in \F \setminus M} \left(\int_{\F \setminus K_{x, \varepsilon}} \L(x,y) \: d\rho(y) + \int_{K_{x, \varepsilon}} \L(x,y) \: d\rho(y)\right) \\
	&\le \varepsilon + \sup_{x \in \F \setminus M} \sup_{y \in K_{x, \varepsilon}}\L(x,y) \: \rho(K_{x, \varepsilon}) < \infty \phantom{\int} 
	\end{align*}
	for some $\varepsilon > 0$ and the corresponding compact subset $K_{x, \varepsilon} \subset \F$ (see \eqref{Kxe}).
\end{proof}

Moreover, under the additional assumptions (b) and (c) in Lemma \ref{Lemma conditions} the following statement is true:

\begin{Corollary}
	Assume that $\L$ is continuous and decays in entropy. 
	Under the additional assumptions~{\rm{(b)} and~\rm{(c)}} in Lemma \ref{Lemma conditions} (again with~$K_x$ replaced by~$K_{x,\varepsilon}$), for all $\varepsilon \in (0,1)$ there is~$\gamma > 0$ such that~$\rho(K_{x, \varepsilon}) \ge \gamma$ for all~$x \in \F$ (where $K_{x, \varepsilon}$ is given by~\eqref{Kxe}).  
\end{Corollary}
\begin{proof}
	Consider an arbitrary $x \in \F$. In view of Theorem \ref{Theorem EL vanishing} we have 
	\begin{align*}
	1 \le \int_{\F \setminus K_{x, \varepsilon}} \L(x,y) \: d\rho(y) + \int_{K_{x, \varepsilon}} \L(x,y) \: d\rho(y) < \varepsilon + \sup_{y \in K_{x, \varepsilon}}\L(x,y) \: \rho(K_{x, \varepsilon})
	\end{align*}
	for some $\varepsilon > 0$ and the corresponding compact set $K_{x, \varepsilon} \subset \F$.
	Choosing $\varepsilon \in (0,1)$, we obtain 
	\begin{align*}
	\rho(K_{x, \varepsilon}) \ge \frac{1 - \varepsilon}{\mathscr{C}} =: \gamma \:,
	\end{align*} 
	which completes the proof.
\end{proof}

\subsection{Existence of Minimizers under Variations of Compact Support}
In the last two subsections we finally return to the question if the measure $\rho$ is a minimizer of the causal variational principle. 
In preparation, we deal with the case of minimizers under variations of compact support (see Definition~\ref{Def global}).

\begin{Thm}%[\textbf{Minimizers under variations of compact support}]
	\label{Lemma minimizer vanishing compact}
	Assume that $\L : \F \times \F \to \R_0^+$ is continuous and decays in entropy (see Definition~\ref{Definition vanishing}). 
	Then $\rho$ is a minimizer under variations of compact support.
\end{Thm}
\begin{proof} 
	Let~$\tilde{\rho} \in \mathfrak{B}_{\F}$ be a variation of compact support. Then~$K := \supp (\tilde{\rho} - \rho )$ is compact, and~$\rho(K) = \tilde{\rho}(K) < \infty$.
	%the difference of the measures~$\tilde{\rho}-\rho$ has compact support. 
%	In view of the fact that 
	Since the Lagrangian is continuous and decays in entropy, the function~$\ell(x)$ (see~\eqref{(Lagrange functional)}) is locally bounded. 
	As a consequence, the difference of the actions~\eqref{integrals} is well-defined. 
	Thus it remains to show that
	\begin{align*}
	\big(\Sact(\tilde{\rho}) - \Sact(\rho) \big) \ge 0
	\end{align*}
	for all variations~$\tilde{\rho}$ of compact support. %For such variations, the set~$K := \supp (\tilde{\rho} - \rho )$ is compact and~$\rho(K) = \tilde{\rho}(K) < \infty$.
	Given~$\tilde{\varepsilon}>0$ and~$x \in K$, we know that
	\[ \int_{\F \setminus K_{x,\tilde{\varepsilon}/2}} \L(x,y) \:d\rho(y) < \frac{{\tilde{\varepsilon}}}{2} \]
	(where $K_{x, \tilde{\varepsilon}/2} \subset \F$ is given according to~\eqref{Kxe}). 
	By continuity of the Lagrangian, there is an open neighborhood~$U_x$ of~$x$ such that
	\[ \int_{\F \setminus K_{x,\tilde{\varepsilon}/2}} \L(z,y) \:d\rho(y) < {\tilde{\varepsilon}} \qquad \text{for all~$z \in U_x$}\:. \]
	%	Without loss of generality, we may assume that $U_x \in \mathscr{D}$ (see \eqref{(D)}). 
	Covering the compact set~$K$ by a finite number of such neighborhoods~$U_{x_1}, \ldots, U_{x_L}$
	and introducing $K_{\tilde{\varepsilon}} := K_{x_1,\tilde{\varepsilon}/2} \cup \cdots \cup K_{x_L,\tilde{\varepsilon}/2}$, we conclude that
	\begin{align}\label{(Ke)}
	\int_{\F \setminus K_{\tilde{\varepsilon}}} \L(x,y) \:d\rho(y) < {\tilde{\varepsilon}} \qquad \text{for all $x \in K$} \:.
	\end{align}
	Similarly, for all $x \in K$ we have 
	\begin{align*}
	\int_{\F \setminus K_{\tilde{\varepsilon}}} \L(x,y) \:d\rho^{(n)}(y) < {\tilde{\varepsilon}} \qquad \text{for sufficiently large $n \in \N$} \:.
	\end{align*} 
	According to \eqref{integrals}, we obtain 
	\begin{align*}
	&\Sact(\tilde{\rho}) - \Sact(\rho) 
	= 2 \int_{K} d(\tilde{\rho} - \rho)(x) \int_{\F \setminus K_{\tilde{\varepsilon}}} d\rho(y) \:\L(x,y) \\
	&\qquad + 2 \int_{K} d(\tilde{\rho} - \rho)(x) \int_{K_{\tilde{\varepsilon}}} d\rho(y) \:\L(x,y) 
	+ \int_{K} d(\tilde{\rho}- \rho)(x) \int_{K} d(\tilde{\rho}- \rho)(y) \:\L(x,y) \:.
	\end{align*}
	Choosing $\tilde{\varepsilon} > 0$ suitably and making use of \eqref{(Ke)}, the expression
	\begin{align*}
	\int_{K} d(\tilde{\rho} - \rho)(x) \int_{\F \setminus K_{\tilde{\varepsilon}}} d\rho(y) \:\L(x,y) \le 2\tilde{\varepsilon} \:\rho(K)
	\end{align*}
	can be arranged to be arbitrarily small, giving rise to 
	\begin{align*}
	\Sact(\tilde{\rho}) - \Sact(\rho) &\ge \left[2\int_{K} d(\tilde{\rho} - \rho)(x) \int_{K_{\tilde{\varepsilon}}} d\rho(y) \:\L(x,y)\right. \\
	&\qquad \qquad \left.+ \int_{K} d(\tilde{\rho}- \rho)(x) \int_{K} d(\tilde{\rho}- \rho)(y) \:\L(x,y)\right] - \varepsilon
	\end{align*}
	for any $\varepsilon > 0$. Proceeding similarly as in the proof of Theorem~\ref{Lemma compact minimizer} and Theorem~\ref{Theorem global}, one can show that the term in square brackets is bigger than or equal to zero, 
	up to an arbitrarily small error term. Since $\varepsilon > 0$ was chosen arbitrarily, we finally arrive at
	\begin{align*}
	\big(\Sact(\tilde{\rho}) - \Sact(\rho)\big) \ge 0 \:,
	\end{align*}
	which proves the claim. 
\end{proof}

\subsection{Existence of Minimizers under Variations of Finite Volume}
Finally we can prove the existence of minimizers in the sense of Definition~\ref{Definition minimizer}. 

\begin{Thm}\label{Theorem minimizer vanishing} 
	Assume that $\L : \F \times \F \to \R_0^+$ is continuous, bounded, and decays in entropy (see Definition \ref{Definition vanishing}). 
	Moreover, assume that condition~{\rm{(iv)}} in~\S\ref{seccvp} holds. Then~$\rho$ is a minimizer under variations of finite volume (see Definition~\ref{Definition minimizer}).
\end{Thm}
\begin{proof}
The idea is to proceed similarly to the proof of Theorem \ref{Theorem global}. Considering variations of finite volume~$\tilde{\rho} \in \mathfrak{B}_{\F}$ satisfying the conditions in~\eqref{totvol} and introducing the set~$B := \supp (\tilde{\rho} - \rho)$, 
we know that~$(\tilde{\rho} - \rho)(B) = 0$ and thus~$\rho(B) = \tilde{\rho}(B) < \infty$. Under the assumption that condition~(iv) in~\S\ref{seccvp} holds, the difference \eqref{integrals} is well-defined, giving rise to 
\begin{align*}
\Sact(\tilde{\rho}) - \Sact(\rho) 
&= 2\int_{B} d(\tilde{\rho} - \rho)(x) \int_{\F} d\rho(y) \:\L(x,y) \\ 
&\qquad \quad + 
\int_{B} d(\tilde{\rho}- \rho)(x) \int_{B} d(\tilde{\rho}- \rho)(y) \:\L(x,y) \:.
\end{align*}
For arbitrary $\tilde{\varepsilon} > 0$, we first approximate $B$ by open sets $U \supset B$ due to regularity of~$\rho$ and $\tilde{\rho}$ such that 
\[\rho(U \setminus B) < \tilde{\varepsilon}/4 \qquad \text{and} \qquad \tilde{\rho}(U \setminus B) < \tilde{\varepsilon}/4 \:, \]
and then we approximate $U$ from inside by compact sets~$K$ such that 
\[\rho(U \setminus K) < \tilde{\varepsilon}/4 \qquad \text{and} \qquad \tilde{\rho}(U \setminus K) < \tilde{\varepsilon}/4 \:. \]
This gives rise to 
\begin{align*}
&\Sact(\tilde{\rho}) - \Sact(\rho) = 2 \int_{B \setminus K} d(\tilde{\rho} - \rho)(x) \int_{\F} d\rho(y) \:\L(x,y) + 2 \int_{K} d(\tilde{\rho} - \rho)(x) \int_{\F} d\rho(y) \:\L(x,y) \\
&\quad + \int_{B \setminus K} d(\tilde{\rho}- \rho)(x) \int_{B} d(\tilde{\rho}- \rho)(y) \:\L(x,y) + \int_{K} d(\tilde{\rho}- \rho)(x) \int_{B} d(\tilde{\rho}- \rho)(y) \:\L(x,y) \:.
\end{align*}
Proceeding similarly to the proof of Theorem \ref{Theorem global} and Theorem~\ref{Lemma minimizer vanishing compact}, one arrives at
\begin{align*}
\Sact(\tilde{\rho}) - \Sact(\rho) &\ge \left[2\int_{K} d(\tilde{\rho} - \rho)(x) \int_{K_{\tilde{\varepsilon}}} d\rho(y) \:\L(x,y)\right. \\
&\qquad \qquad \left.+ \int_{K} d(\tilde{\rho}- \rho)(x) \int_{K} d(\tilde{\rho}- \rho)(y) \:\L(x,y)\right] - \varepsilon
\end{align*}
for any $\varepsilon > 0$ by suitably choosing $U \supset B$ and $K \subset U$. Applying the same arguments as in the proof of Theorem~\ref{Lemma minimizer vanishing compact}, one can show that the term in square brackets is bigger than or equal to zero, up to an arbitrarily small error term. Since $\varepsilon > 0$ was chosen arbitrarily, we finally arrive at
\begin{align*}
\big(\Sact(\tilde{\rho}) - \Sact(\rho)\big) \ge 0 \:,
\end{align*}
which proves the claim. 
\end{proof}

Theorem \ref{Theorem minimizer vanishing} concludes the existence theory in the $\sigma$-locally compact setting. 

\appendix

\section{Non-Triviality of the Constructed Measure}\label{Appendix nontrivial}
The following results show that the measure $\rho$ given by \eqref{rhoglobal} is non-zero in the case that the Lagrangian $\L : \F \times \F \to \R_0^+$ is either of compact range or is continuous and decays in entropy. 

\begin{Lemma}\label{Lemma nontrivial}
	Assume that the Lagrangian~$\L : \F \times \F \to \R_0^+$ is of compact range. Then the measure~$\rho$ obtained in \eqref{rhoglobal} is non-zero. The total volume $\rho(\F)$ is possibly infinite. 
\end{Lemma}
\begin{proof}
	Let $(K_n)_{n \in \N}$ be the compact exhaustion of $\F$ with $K_1 \not= \varnothing$ and $K_n \subset K_{n+1}^{\circ}$ for all $n \in \N$. By construction of $\rho^{(n)}$ we are given $\supp \rho^{(n)} \subset K_{n}$ for every $n \in \N$. Assuming that $\L$ is of compact range, for every $x \in \F$ there is~$K_x \subset \F$ compact such that $\L(x,y) = 0$ for all $y \notin K_x$, and $K_x \subset K_n^{\circ}$ for sufficiently large $n \in \N$.
As a consequence, 
	\begin{align*}
	1 \le \int_{K_x} \L(x,y) \: d\rho^{(n)}(y) \le \sup_{y \in K_x} \L(x,y) \: \rho^{(n)}(K_x) \:,
	\end{align*}
showing that~$\rho^{(n)}(K_x) \ge c_x$ for some constant $c_x > 0$ for sufficiently large $n \in \N$.
Moreover, choosing $f \in C_c(\F; [0,1])$ with $f|_{K_x} \equiv 1$, applying
	vague convergence~\eqref{(vague convergence)} yields
	\begin{align*}
	\int_{\F} f \: d\rho \stackrel{\eqref{(vague convergence)}}{=} \lim_{n \to \infty} \int_{\F} f \: d\rho^{(n)} \ge \lim_{n \to \infty} \rho^{(n)}(K_x) \ge c_x > 0 \:.
	\end{align*} 
We conclude that the measure is non-zero.
	
	Whenever $(K_{x_j})_{j \in \N}$ is a disjoint sequence of such compact sets, we may choose open sets $U_j \supset K_{x_j}$ such that $U_i \cap U_j = \varnothing$ for all~$i \not= j$. Choosing $f_j \in C_c(U_j;[0,1])$ with~$f_j|_{K_{x_j}} \equiv 1$ for every $j \in \N$, we obtain
	\begin{align*}
	\rho(\F) = \int_{\F} d\rho \ge \sum_{j \in \N} \int_{U_j} f_j \: d\rho = \sum_{j \in \N} \lim_{k \to \infty} \int_{U_j} f_j \: d\rho^{(k)} \ge \sum_{j \in \N} \lim_{k \to \infty} \rho^{(k)}(K_{x_j}) \ge \sum_{j \in \N} c_{x_j} \:.
	\end{align*}
	Hence the total volume $\rho(\F)$ is infinite if $\sum_{j \in \N} c_{x_j}$ diverges. 
\end{proof}

\begin{Lemma}\label{Lemma nontrivial decay}
	Assume that the Lagrangian~$\L : \F \times \F \to \R_0^+$ is continuous and decays in entropy. 
	Then the measure~$\rho$ obtained in \eqref{rhoglobal} is non-zero, and the total volume~$\rho(\F)$ is possibly infinite. 
\end{Lemma}
\begin{proof}
	For arbitrary $x \in \F$ and $0 < \varepsilon < 1$ let
	\begin{align*}
	\tilde{\L}_{x, \varepsilon} : \F \to \R_0^+ \:, \qquad \tilde{\L}_{x, \varepsilon}(y) := \L(x,y) \: \chi_{K_{x, \varepsilon}} 
	\end{align*}
	with $K_{x, \varepsilon} \subset \F$ compact given by \eqref{(less epsilon)}. In particular, $\tilde{\L}_{x, \varepsilon}(y) = 0$ for all $y \notin K_{x, \varepsilon}$. 
	As a consequence, $\rho^{(n)}(K_{x, \varepsilon}) \ge c_{x, \varepsilon}$ for some constant $c_{x, \varepsilon} > 0$ for sufficiently large $n \in \N$, as the following argument shows. Assuming conversely that $\rho^{(n)}(K_{x, \varepsilon}) \to 0$ as $n \to \infty$, in view of \eqref{(less epsilon)} we arrive at the contradiction
	\begin{align*}
	1 \le \int_{\F \setminus K_{x, \varepsilon}} \L(x,y) \: d\rho^{(n)}(y) + \int_{K_{x, \varepsilon}} \tilde{\L}_{x, \varepsilon}(y) \: d\rho^{(n)}(y) \le \varepsilon + \sup_{y \in K_{x, \varepsilon}} \L(x,y) \: \rho^{(n)}(K_{x, \varepsilon}) < 1
	\end{align*}
	for sufficiently large $n \in \N$. Choosing $f \in C_c(\F; [0,1])$ with $f|_{K_{x,\varepsilon}} \equiv 1$, 
	according to vague convergence~\eqref{(vague convergence)} %and $K_{x,\varepsilon} \subset K_n$ for sufficiently large $n \in \N$ 
	we obtain
	\begin{align*}
	\int_{\F} f \: d\rho \stackrel{\eqref{(vague convergence)}}{=} \lim_{n \to \infty} \int_{\F} f \: d\rho^{(n)} \ge \lim_{n \to \infty} \rho^{(n)}(K_{x, \varepsilon}) \ge c_{x, \varepsilon} > 0 \:.
	\end{align*} 
	Thus the measure~$\rho$ is non-zero. The last assertion can be proven analogously to the proof of Lemma~\ref{Lemma nontrivial}. 
\end{proof}

\Thanks {{\em{Acknowledgments:}}
We would like to thank Magdalena Lottner, Marco Oppio, Johannes Wurm and the unknown referee for helpful comments on the manuscript.
C.~L.\ gratefully acknowledges generous support by the ``Studienstiftung des deutschen Volkes.''

%\bibliographystyle{amsplain}
%\bibliography{../felix}
\providecommand{\bysame}{\leavevmode\hbox to3em{\hrulefill}\thinspace}
\providecommand{\MR}{\relax\ifhmode\unskip\space\fi MR }
\providecommand{\MRhref}[2]{%
  \href{http://www.ams.org/mathscinet-getitem?mr=#1}{#2}
}
\providecommand{\href}[2]{#2}

\end{document}